\RequirePackage{fix-cm}
\documentclass[twocolumn]{svjour3}          
\smartqed  
\usepackage{amsmath,amssymb}
\usepackage{mathrsfs}
\usepackage[latin1]{inputenc}
\usepackage{graphicx}
\usepackage{xspace}
\usepackage{multirow,enumitem}
\usepackage{booktabs}
\usepackage{microtype,mathtools}
\usepackage[lined]{algorithm2e}
\usepackage{tikz}
\usetikzlibrary{decorations,decorations.pathreplacing,calc}

\tikzset{%
  bigdots/.style={circle, draw=black, fill=black, inner sep=0pt, minimum size=1pt},
  descr/.style={fill=white,inner sep=2.5pt}
}
\newcommand{\unit}{30pt}

\newcounter{cclaim}[section]
\renewcommand{\thecclaim}{\thesection.\arabic{cclaim}}
\newenvironment{observ}{\refstepcounter{cclaim}\par\medbreak\noindent{\bfseries Observation \thecclaim}\itshape}%
{\par\noindent}
\numberwithin{cclaim}{section}
\newcommand{\card}[1]{\left|{#1}\right|}

\newcommand{\set}[1]{\left\{{#1}\right\}}

\renewcommand{\le}{\leqslant}
\renewcommand{\leq}{\leqslant}
\renewcommand{\ge}{\geqslant}
\renewcommand{\geq}{\geqslant}
\newcommand{\abs}[1]{\left\lvert#1\right\rvert}

\newcommand{\lo}[1]{o\mathopen{}(#1)}
\newcommand{\bigo}[1]{O\mathopen{}(#1)}
\newcommand{\thet}[1]{\theta\mathopen{}(#1)}
\newcommand{\intervalle}[4]{\mathopen{#1}#2\mathclose{}\mathpunct{},#3\mathclose{#4}}
\newcommand{\icc}[2]{\intervalle{[}{#1}{#2}{]}}
\newcommand{\ioc}[2]{\intervalle{(}{#1}{#2}{]}}

\newcommand{\ioo}[2]{\intervalle{(}{#1}{#2}{)}}
\newcommand{\sst}[2]{\left\{#1\,:\,#2\right\}}

\newcommand{\LOCAL}{{$\cal LOCAL$}}

\newcommand{\id}{\text{Id}}

\newcommand{\inp}{\mathbf{x}}
\newcommand{\out}{\mathbf{y}}

\newcommand{\hout}{\hat{\out}}

\newcommand{\Reals}[0]{\mathbf{R}}
\newcommand{\Natural}[0]{\mathbf{N}}
\newcommand{\Prob}{\Pi}

\newcommand{\MIS}{\text{\texttt{MIS}}}
\newcommand{\mm}{\text{\texttt{MM}}}

\newcommand{\parameter}{\text{\texttt{p}}}
\newcommand{\para}{\text{\texttt{p}}}
\newcommand{\q}{\text{\texttt{q}}}
\newcommand{\tp}{\tilde{\text{\texttt{p}}}}

\newcommand{\cF}{{\cal F}}
\newcommand{\cA}{{\cal A}}
\newcommand{\cB}{{\cal B}}
\newcommand{\cC}{{\cal C}}

\newcommand{\cP}{{\cal P}}
\newcommand{\cU}{{\cal U}}

\DeclareMathOperator{\dist}{dist}
\DeclareMathOperator{\run}{T}
\DeclareMathOperator{\Ee}{\mathbf{E}}
\DeclareMathOperator{\pr}{\mathbf{Pr}}

\journalname{Distributed Computing}
\begin{document}

\title{Toward More Localized Local Algorithms:\\ Removing
  Assumptions Concerning Global Knowledge\thanks{%
Amos Korman is supported in part by a  France-Israel cooperation grant
 (``Mutli-Computing'' project)
 from the France Ministry of Science and Israel Ministry of Science,
 by the ANR projects ALADDIN and PROSE, and by the INRIA project
 GANG.
Jean-Sébastien Sereni is partially
 supported by the French \emph{Agence Nationale de la Recherche} under reference
 \textsc{anr 10 jcjc 0204 01}.
Laurent Viennot is supported by the european STREP project EULER, and the INRIA project-team  GANG.
}}



 \author{%
 Amos Korman
 \and
 Jean-Sébastien Sereni
 \and
 Laurent Viennot
 }


 \institute{%
 Amos Korman \at
   CNRS and University Paris Diderot\\
   LIAFA Case 7014\\
   Université Paris Diderot -- Paris 7\\
   F-75205 Paris Cedex 13, France.\\
   Tel.: +33-1-57-27-92-56\\
   Fax: +33-1-57-27-94-09\\
   \email{Amos.Korman@liafa.jussieu.fr}
 \and
 Jean-Sébastien Sereni \at
   CNRS (LIAFA, Universit\'e Denis Diderot), Paris, France\\
   and Department of Applied Mathematics (KAM), Faculty of Mathematics and
   Physics, Charles University, Prague, Czech Republic\\
   \email{sereni@kam.mff.cuni.cz}
 \and
 Laurent Viennot \at
   INRIA and University Paris Diderot\\
   LIAFA Case 7014\\
   Université Paris Diderot -- Paris 7\\
   F-75205 Paris Cedex 13, France.\\
   \email{Laurent.Viennot@inria.fr}
 }

\date{Received: date / Accepted: date}

\maketitle

\begin{abstract}
Numerous sophisticated local algorithm were suggested in the
literature for various fundamental problems. Notable examples are
the MIS and $(\Delta+1)$-coloring algorithms by Barenboim and
Elkin~\cite{BaEl10b}, by Kuhn~\cite{Kuh09}, and by Panconesi and
Srinivasan~\cite{PaSr96}, as well as the $\bigo{\Delta^2}$-coloring
algorithm by Linial~\cite{Lin92}. Unfortunately, most known local
algorithms (including, in particular, the aforementioned algorithms)
are \emph{non-uniform}, that is, local algorithms generally use
good estimations of one or more global parameters of the network,
e.g., the maximum degree $\Delta$ or the number of nodes $n$.

This paper provides a method for transforming a
non-uniform local algorithm into a \emph{uniform} one. Furthermore,
the resulting algorithm enjoys the same asymptotic running time as
the original non-uniform algorithm. Our method applies to a wide
family of both deterministic and randomized algorithms.
Specifically, it applies to almost all state of the art
non-uniform algorithms for MIS and Maximal Matching, as well
as to many results concerning the coloring problem.  (In particular,
it applies to all aforementioned algorithms.)

To obtain our transformations we introduce a new distributed tool
called \emph{pruning algorithms}, which we believe may be of
independent interest.

\end{abstract}


\section{Introduction}

\subsection{Background and Motivation}

Distributed computing concerns  environments in which many processors,
located at different sites, must collaborate in order to achieve some global
task.  One of the main themes in distributed network algorithms concerns~the
question of how to cope with \emph{locality} constrains, that is, the lack
of knowledge about the global structure of the network (cf.,~\cite{Pel00}).
On the one hand, information about the global structure may not always be
accessible~to~individual processors and the cost of computing it from
scratch  may overshadow the cost of the algorithm using it. On the other
hand, global knowledge is not always essential, and many seemingly global
tasks can be efficiently achieved by letting processors know more about
their immediate neighborhoods and less about the rest of the network.

A standard model for capturing the essence of locality is the $\cal{LOCAL}$
model (cf.,~\cite{Pel00}).  In this model, the network is modeled by a graph
$G$, where the nodes of $G$ represent the processors and the edges
represent the communication links.  To perform a task, nodes are woken up
simultaneously, and computation proceeds in fault-free synchronous rounds
during which every node exchanges messages with its neighbors, and performs
arbitrary computations on its data.  Since many tasks cannot be solved
distributively in an anonymous network in a deterministic way,
symmetry breaking must be addressed. Arguably, there are two typical ways
to address this issue: the first one is to use randomized algorithms,
while the second one is to assume that each node $v$ in the network
is initially provided a unique identity $\id(v)$.
A \emph{local algorithm} operating in such a setting must return an output
at each node such that the collection of outputs satisfies the required
task.  For example, a \emph{Maximal Independent Set} (MIS) of a graph $G$
is a set $S$ of nodes of $G$ such that every node not in $S$ has a neighbor in $S$
and no two nodes of $S$ are adjacent.
In a local algorithm for the MIS problem, the output
at each node $v$ is a bit $b(v)$ indicating whether $v$ belongs to a
selected set $S$ of nodes, and it is required that $S$ forms a
MIS of $G$.  The \emph{running time} of a local algorithm  is the number
of rounds needed for the algorithm to complete its operation at each node,
taken in the worst case scenario.  This is typically evaluated with respect
to some parameters of the underlying graph.  The common parameters used are
the number  of nodes $n$ in the graph and the maximum degree $\Delta$ of a
node in the graph.

To ease the computation, it is often assumed that some kind of knowledge
about the global network is provided to each node \emph{a priori}. A typical
example of such knowledge is the number of nodes $n$ in the network.  It
turns out that in some cases, this (common) assumption can give a lot of
power to the distributed algorithm.  This was observed by Fraigniaud
\emph{et al.}~\cite{FKP11} in the context of local decision: they introduced
the complexity class of decision problems NLD, which contains all decision
problems that can be verified in constant time with the aid of a
certificate. They proved that, although there exist decision problems that
do not belong to NLD, every (computable) decision problem falls in NLD if it
is assumed that each node is given the value of $n$ as an input.

In general, the amount and type of such information may have a profound
effect on the design of the distributed algorithm.  Obviously, if the whole
graph is contained in the input of each node, then the
distributed algorithm can be reduced to a central one.  In fact, the whole
area of \emph{computation with advice} \cite{CFI+08,DePe10,FGIP09,FIP10,FKL07,KoKu07,KKP10}
is dedicated to studying the amount of information contained in the inputs
of the nodes and its effect on the performances of the distributed algorithm. For instance,
Fraigniaud \emph{et al.}~\cite{FKL07} showed that if each node is provided
with only a constant number of bits then one can locally construct a
BFS-tree in constant time, and  can locally construct a MST in $\bigo{\log
n}$ time, while both tasks require diameter time if no knowledge is assumed.
As another example, Cohen \emph{et al.}~\cite{CFI+08} proved that $\bigo{1}$
bits, judiciously chosen at each node, can allow a finite automaton to
distributively explore every graph.  As a matter of fact, from a radical
point of view, for many questions (e.g., MIS and Maximal Matching),
additional information may push the question at hand into absurdity: even a
constant number of bits of additional information per node is enough to
compute a solution---simply let the additional information encode the
solution!

When dealing with locality issues, it is desired that the amount of
information regarding the whole network contained in the inputs of the nodes is
minimized.  A local algorithm that assumes that each node is initially given
merely its own identity is often called \emph{uniform}. Unfortunately, there
are only few local algorithms in the literature that are uniform (e.g.,
\cite{BGPV08,LOW08,LPR09,Lub86,ScWa10b}).  In contrast, most known local
algorithms  assume that the inputs of all nodes contain upper bounds on the values of some
global parameters of the network.  Moreover, it is often assumed that all
inputs contain the same upper bounds on the global parameters.  Furthermore,
typically, not only the correct operation of the algorithm requires that
upper bounds be contained in the inputs of all nodes,
but also the running time of the algorithm
is actually a function of the upper bound estimations and not of the actual values of the
parameters. Hence, it  is desired that the upper bounds contained in the
inputs are not significantly larger than the real values of the parameters.

Some attempts to transform a non-uniform local algorithm into a uniform one
were made by examining the details of the algorithm at hand and modifying it
appropriately. For example, Barenboim and Elkin~\cite{BaEl10b} first gave a
non-uniform MIS algorithm for the family of graphs with arboricity
$a=\bigo{{\log^{1/2-\delta} n}}$, for any constant $\delta\in (0,1/2)$,
running in time $\bigo{\log n/\log\log n}$. (The \emph{arboricity} of a graph
being the smallest number of acyclic subgraphs that together contain all the edges
of the graph.) At the cost of increasing the
running time to $\bigo{{\frac{\log n}{\log\log n}\log^* n}}$, the authors
show how to modify their algorithm so that the value of $a$ need not be part
of the inputs of nodes.  In addition to the MIS algorithms, the work of \cite{BaEl10b} also contains
algorithms that do not require the knowledge of the arboricity, but have
the same asymptotic running time as the ones that require it. For example, this corresponds to algorithms
computing forests-decomposition and $O(a)$-coloring.
Nevertheless, all their algorithms still require the inputs of all nodes to contain
a common upper bound on $n$.

We present general methods for transforming a
non-uniform local algorithm into a uniform one without increasing the
asymptotic running time of the original algorithm. Our method applies to
a wide family of both deterministic and randomized algorithms. In
particular, our method applies to all state of the art non-uniform
algorithms for MIS and Maximal Matching, as well as to several of the
best known results for $(\Delta+1)$-coloring.

Our transformations are  obtained using a new type of local algorithms
termed \emph{pruning algorithms}. Informally, the basic property of a
pruning algorithm is that it allows one to iteratively apply a sequence of
local algorithms (whose output may not form a correct global solution) one
after the other, in a way that ``always progresses'' toward a  solution.  In a
sense, a pruning algorithm   is a combination of a gluing mechanism and a
\emph{local checking} algorithm (cf., \cite{FKP11,NaSt95}). A local checking
algorithm for a problem $\Prob$ runs on graphs with an output value at each
node (and possibly an input too), and can locally detect whether the output
is ``legal'' with respect to $\Prob$.  That is, if  the instance is not
legal then at least one node detects this, and raises an alarm.  (For
example, a local checking algorithm for MIS is trivial: each node in the set
$S$, which is suspected to be a MIS, checks that none of its neighbors
belongs to  $S$, and each node not in $S$ checks that at least one of its
neighbors belongs to  $S$. If the check fails, then the node raises an
alarm.) A pruning algorithm needs to satisfy an additional \emph{gluing}
property not required by local checking algorithms.  Specifically, if the
instance is not legal, then the pruning algorithm must carefully choose the
nodes raising the alarm (and possibly modify their input too), so that a
solution for the subgraph induced by those alarming nodes can be well glued
to the previous output of the non-alarming nodes, in a way such  that the
combined output is a solution to the problem for the whole initial graph.

We believe that  this new  type of algorithms may be of independent
interest.  Indeed, as we show, pruning algorithms have several types of
other applications in the theory of local computation, besides the
aforementioned issue of designing uniform algorithms. Specifically, they can
be used also to transform a local Monte-Carlo algorithm into a Las Vegas
one, as well as to obtain an algorithm that runs in the minimum running time
of a given (finite) set of uniform algorithms.

\subsection{Previous Work}

\begin{table*}[!t]
\begin{center}
\scalebox{0.9}{
\bgroup\small
\begin{tabular}{@{}lccccr@{}}\toprule
Problem & Parameters & Time & Ref. & This paper (uniform) & Corollary~\ref{cor:main}\\
\midrule
Det. MIS and $(\Delta\negthickspace+\negthickspace1)$-coloring
    & $n,\Delta$ & $\bigo{\Delta + \log^*n}$ & \cite{BaEl09,Kuh09}
    & \multirow{2}{5cm}{$\min\set{\bigo{\Delta + \log^*n},
        2^{\bigo{\sqrt{\log n}}}}$} & \ref{thm:mis} \\
  & $n$ & $2^{\bigo{\sqrt{\log n}}}$ & \cite{PaSr96} & & \ref{thm:color-delta}\\

Det. MIS (arboricity $a=\lo{\sqrt{\log n}}$)
   & $n,a$ & $\lo{\log n}$ & \cite{BaEl10b}
   & $\lo{\log n}$ & \ref{thm:mis} \\

Det. MIS (arboricity $a=\bigo{{\log^{1/2-\delta}n}}$)
   & $n,a$ & $\bigo{\log n/\log\log n}$ & \cite{BaEl10b}
   & $\bigo{\log n/\log\log n}$ & \ref{thm:mis} \\

Det. $\lambda(\Delta+1)$-coloring
  & $n,\Delta$ & $\bigo{\Delta/\lambda+\log^* n}$ & \cite{BaEl09,Kuh09}
  & $\bigo{\Delta/\lambda+\log^* n}$ & \ref{thm:uniformtradeoff}\\

Det. $\bigo{\Delta}$-edge-coloring
  & $n,\Delta$ & $\bigo{\Delta^\epsilon+\log^* n}$ & \cite{BaEl11}
  & $\bigo{\Delta^\epsilon+\log^* n}$ & \ref{thm:edge-coloring}\\

Det. $\bigo{\Delta^{1+\epsilon}}$-edge-coloring
  & $n,\Delta$ & $\bigo{\log \Delta+\log^* n}$ & \cite{BaEl11}
  & $\bigo{\log \Delta+\log^* n}$ & \ref{thm:edge-coloring}\\

Det. Maximal Matching
  & $n$ or $\Delta$ & $\bigo{\log^4 n}$ & \cite{HKP01}
  & $\bigo{\log^4 n}$ & \ref{thm:mm}\\

Rand. $(2,2(c+1))$-ruling set
  & $n$ & $\bigo{2^c\log^{1/c}n}$ & \cite{ScWa10a}
  & $\bigo{2^c\log^{1/c}n}$ & \ref{thm:rs}\\

Rand. MIS
& uniform & $\bigo{\log n}$ & \cite{ABI86,Lub86}
  & & \\
\bottomrule
\end{tabular}
\egroup
}
\end{center}
\label{tab:results}
\caption{Comparison of \LOCAL{} algorithms with respect to the use
of global parameters. ``Det.'' stands for deterministic, and
  ``Rand.'' for randomized.}
\end{table*}

\paragraph{MIS and coloring:}
There is a long line of research concerning the two related problems of $(\Delta+1)$-coloring and
MIS~\cite{AGLP89,CoVi86,GoPl87,GPS88,KMW04,KuWa06,Lin92}.
A \emph{$k$-coloring} of a graph is an assignment of an integer in
$\{1,\ldots,k\}$ to each node such that
no two adjacent vertices are assigned the same integer.
Recently, Barenboim and Elkin~\cite{BaEl09}
and independently Kuhn~\cite{Kuh09} presented two elegant
$(\Delta+1)$-coloring and MIS algorithms running in $\bigo{\Delta+\log^* n}$ time on general graphs. This
is the best currently-known bound for these problems on low degree graphs. For graphs with a
large maximum degree $\Delta$, the best
bound is due to Panconesi and
Srinivasan~\cite{PaSr96}, who devised an algorithm running in $2^{\bigo{\sqrt{\log n}}}$ time.
The aforementioned algorithms are not uniform.
Specifically, all three algorithms
require that the inputs of all nodes contain a common upper
bound on $n$ and the first two also require a common upper bound on $\Delta$.

For bounded-independence graphs, Schneider and Wattenhofer~\cite{ScWa10b} designed uniform
deterministic MIS and $(\Delta+1)$-coloring algorithms  running in
$\bigo{\log^* n}$ time.  Barenboim and Elkin~\cite{BaEl10b} devised a
deterministic algorithm for the MIS problem on graphs of bounded arboricity
that requires time $\bigo{{\log n/\log\log n}}$.  More specifically, for graphs
with arboricity $a=\lo{\sqrt{\log n}}$,  they show that a MIS can be computed
deterministically in $\lo{\log n}$ time, and whenever $a=\bigo{{\log^{1/2-\delta}
n}}$ for some constant $\delta\in\ioo{0}{1/2}$, the same algorithm runs in time
$\bigo{\log n/\log\log n}$.  At the cost of increasing the running time by a
multiplicative factor of $\bigo{\log^* n}$, the authors show how to
modify their algorithm so that the value of $a$ need not be part of the inputs
of nodes.  Nevertheless,
all their algorithms require the inputs of all nodes to contain a common upper bound
on the value of $n$.
Another MIS algorithm which is efficient for graphs with low arboricity was devised by Barenboim and Elkin~\cite{BaEl10a}; this algorithm runs in time
$\bigo{a+a^{\epsilon}\log n}$ for arbitrary constant $\epsilon>0$.

Concerning the problem of coloring with more than $\Delta+1$ colors,
Linial~\cite{Lin87,Lin92}, and subsequently Szegedy and Vishwanathan~\cite{SzVi93},
described $\bigo{\Delta^2}$-coloring algorithms with running time $\thet{\log^* n}$.
Barenboim and  Elkin~\cite{BaEl09} and, independently, Kuhn~\cite{Kuh09}
generalized this by presenting
a tradeoff between the running time and the number of colors: they devised
a $\lambda(\Delta+1)$-coloring algorithm with running
time $\bigo{\Delta/\lambda+\log^* n}$, for any $\lambda\geq 1$.
All these algorithms
require the inputs of all nodes to contain common upper bounds on both $n$ and $\Delta$.

Barenboim and Elkin~\cite{BaEl10a} devised a $\Delta^{1+\lo{1}}$ coloring algorithm
running in time  $\bigo{f(\Delta) \log \Delta  \log n}$, for
an arbitrarily slow-growing function $f = \omega(1)$. They
also produced an $\bigo{\Delta ^{1+\epsilon}}$-coloring algorithm running
in $\bigo{\log \Delta  \log n}$-time, for
an arbitrarily small constant $\epsilon > 0$, and an $\bigo{\Delta}$-coloring algorithm running in
$\bigo{\Delta^{\epsilon} \log n}$ time, for an arbitrarily small constant $\epsilon> 0$.
All these coloring algorithms require the inputs of all nodes to contain the
values of both $\Delta$ and $n$.
Other deterministic non-uniform coloring algorithms with number of colors and
running time~corresponding to the arboricity parameter were given by
Barenboim and Elkin~\cite{BaEl10a,BaEl10b}.

Efficient deterministic algorithms for the edge-coloring
problem can
be found in several papers~\cite{BaEl10a,BaEl11,PaRi01}.  In particular, Panconesi and Rizzi~\cite{PaRi01} designed a simple
deterministic local algorithm that finds a $(2\Delta-1)$-edge-coloring of a
graph in time $\bigo{\Delta+\log^* n}$.
Recently, Barenboim and Elkin~\cite{BaEl11},
designed an $\bigo{\Delta}$-edge-coloring algorithm running in
time $\bigo{\Delta^\epsilon}+\log^* n$, for any $\epsilon>0$, and an
$\bigo{\Delta^{1+\epsilon}}$-edge-coloring algorithm running in time
$\bigo{\log \Delta}+\log^* n$, for any $\epsilon>0$.
All these
algorithms require the inputs of all nodes to contain common upper bounds on both $n$ and
$\Delta$.

Randomized algorithms for  MIS and $(\Delta+1)$-coloring
running in expected time $\bigo{\log n}$ were initially given by
Luby~\cite{Lub86} and, independently, by Alon \emph{et al.}~\cite{ABI86}.

Recently, Schneider and Wattenhofer~\cite{ScWa10a} constructed the best
currently-known
non-uniform $(\Delta+1)$-coloring algorithm, which runs in
time
$\bigo{\log \Delta +\sqrt{\log n}}$.
They also provided random algorithms for coloring using more colors.
For every positive integer $c$,
a randomized algorithm for $(2,2(c+1))$-ruling set
running in time $\bigo{2^c\log^{1/c}n}$ is also presented.
(A set $S$ of nodes in a graph being \emph{$(\alpha,\beta)$-ruling}
if every node not in $S$ is at distance at most $\beta$ of a node in $S$
and no two nodes in $S$ are at distance less than $\alpha$.)
All these algorithms of Schneider and Wattenhoffer~\cite{ScWa10a} are not uniform and require
the inputs of all nodes to contain a common upper bound on $n$.

\paragraph{Maximal Matching:} A \emph{maximal matching} of a graph $G$ is a
set $M$ of edges of $G$ such that every edge not in $M$ is incident to an edge
in $M$ and no two edges in $M$ are incident.
Schneider and Wattenhofer~\cite{ScWa10b} designed a uniform deterministic maximal matching
algorithm on bounded-independence graphs running in $\bigo{\log^* n}$ time.
For general graphs, however, the state of the art maximal matching algorithm
is not uniform: Hanckowiak \emph{et al.}~\cite{HKP01} presented a non-uniform deterministic
algorithm for maximal matching running in time $\bigo{\log^4 n}$.
This algorithm assumes that the inputs of all nodes contain a common upper
bound on $n$
(this assumption can be omitted for some parts of the algorithm under
the condition that the inputs of all nodes contain the value of $\Delta$).

\subsection{Our Results}
The main conceptual contribution of the paper is the introduction of a new
type of algorithms called \emph{pruning algorithms}. Informally, the fundamental property
of this type of algorithms is to allow one to iteratively run a sequence of algorithms
(whose output may not necessarily be correct everywhere)  so that the global output does not deteriorate, and  it always progresses
toward a solution.

Our main application for pruning algorithm  concerns the problem of locally computing a
global solution while minimizing the necessary global information contained
in the inputs of the nodes.  Addressing this, we  provide a method for
transforming a non-uniform local algorithm into a uniform one without
increasing the asymptotic running time of the original algorithm. Our method
applies to a wide family of both deterministic and randomized algorithms; in
particular, it applies to many of the best known results
concerning classical problems such as MIS, Coloring, and Maximal Matching.
(See Table \ref{tab:results} for a summary of some of the uniform algorithms
we obtain and the corresponding state of the art existing non-uniform
algorithms.)

In another application, we show how to transform a Monte-Carlo local
algorithm into a Las Vegas one.
Finally, given a constant number of uniform algorithms for
the same problem whose running times depend on different parameters---which
are unknown to nodes---we show a method for constructing a uniform
algorithm solving the problem, that on every instance runs asymptotically as
fast as the fastest algorithm among those given algorithms.

Stating our main results requires a number of formal definitions, so we defer
the precise statements to later parts of the paper. Rather, we provide here
some interesting corollaries of our results. References for the
corresponding non-uniform algorithms are provided in Table~\ref{tab:results}.
(The notion of
``moderately-slow function'' used in item~\ref{thm:uniformtradeoff} below is
defined in Section~\ref{sec:preliminaries}.)

\begin{corollary}\label{cor:main}
\mbox{}
\begin{enumerate}[label=\textup{(\roman{*})}, ref={(\roman{*})}]
\item\label{thm:mis}
There exists a uniform deterministic algorithm solving $\MIS$ on general
graphs in time
\[
\min\left\{g(n),h(\Delta,n), f(a,n)\right\},
\]
where  $g(n)=2^{\bigo{\sqrt{\log n}}}$, $h(\Delta,n)=\bigo{\Delta+\log^* n}$, and $f(a,n)$ is bounded as follows. $f(a,n)=\lo{\log n}$ for graphs of arboricity $a=\lo{\sqrt{\log n}}$,
$f(a,n)=\bigo{\log n/\log\log n}$ for arboricity
$a=\bigo{\log^{1/2-\delta} n}$, for some constant $\delta\in\ioo{0}{1/2}$; and
otherwise: $f(a,n)=\bigo{a+a^{\epsilon}\log n}$, for arbitrary small constant $\epsilon>0$.\\

\item\label{thm:color-delta}
There exists a uniform deterministic algorithm solving the $(\Delta+ 1)$-coloring  problem  on general graphs in time $\min\{\bigo{\Delta+\log^* n},2^{\bigo{\sqrt{\log n}}}\}$.\\ 

\item\label{thm:uniformtradeoff}
There exists a uniform deterministic algorithm solving the
$\lambda(\Delta+1)$-coloring problem on general graphs and running in time
$\bigo{\Delta/\lambda+\log^* n}$, for any $\lambda\geq 1$, such that $\Delta/\lambda$ is a
moderately-slow function.
In particular, there exists a uniform deterministic algorithm solving the
$\bigo{\Delta^2}$-coloring problem in time $\bigo{\log^* n}$.\\

\item\label{thm:uniform-sublinear}
The following uniform deterministic coloring algorithms exist.\\
\begin{itemize}
\item
A uniform $\Delta^{1+o(1)}$-coloring algorithm
running in time  $\bigo{f(\Delta) \log \Delta  \log n\log\log n}$, for
an arbitrarily slow-growing function $f = \omega(1)$.\\
\item
A uniform
 $\bigo{\Delta ^{1+\epsilon}}$-coloring algorithm running  in $\bigo{\log
 \Delta  \log n\log\log n}$ time, for any constant $\epsilon>0$.\\
\item
A uniform $\bigo{\Delta}$-coloring algorithm running in\break
$\bigo{\Delta^{\epsilon} \log n \log\log n}$ time, for any constant $\epsilon>0$.\\
\end{itemize}

\item\label{thm:edge-coloring}
\begin{itemize}
\item There exists a uniform deterministic
$\bigo{\Delta}$-edge-coloring algorithm  for general graphs running in time
$\bigo{\Delta^\epsilon+\log^* n}$, for any constant $\epsilon>0$.\\
\item There exists a
uniform deterministic $\bigo{\Delta^{1+\epsilon}}$-edge-coloring algorithm
for general graphs that runs in time
$\bigo{\log \Delta+\log^* n}$, for any constant $\epsilon>0$.\\
\end{itemize}

\item\label{thm:mm}
There exists a uniform deterministic algori\-thm  solving the
maximal matching problem in time $\bigo{\log^4\! n}$.\\

\item\label{thm:rs}
For a constant integer $c\ge 1$, there exists a uniform randomized algorithm
solving the $(2,2(c+1))$-ruling set problem in time $\bigo{2^c\log^{1/c}n}$.
\end{enumerate}
\end{corollary}

\section{Preliminaries}\label{sec:preliminaries}
\paragraph{General definitions:}
For two integers $a$ and $b$, we let
$\icc{a}{b}=\{a,a+1,\ldots,b\}$.
A vector $\underline{x}\in\Reals^\ell$ is said
to \emph{dominate} a vector $\underline{y}\in\Reals^\ell$ if $\underline{x}$
is coordinate-wise greater than or equal to $\underline{y}$, that is,
$\underline{x}_k\ge\underline{y}_k$ for each $k\in\icc{1}{\ell}$.

 For a graph $G$, we let $V(G)$ and $E(G)$ be the sets of nodes and
edges of $G$, respectively. (Unless mentioned otherwise, we consider only undirected and unweighted graphs.) The \emph{degree}
$\deg_G(v)$ of a node $v\in V(G)$ is the number of neighbors of $v$ in $G$.
The \emph{maximum degree of $G$} is $\Delta_G=\max\sst{\deg_G(v)}{v\in V(G)}.$

Let $u$ and $v$ be two nodes of $G$.
The \emph{distance} $\dist_G(u,v)$ between $u$ and $v$ is the number of
edges on a shortest path connecting them. 
Given an integer $r\ge 0$, the \emph{ball} of radius $r$
around $u$ is the subgraph $B_G(u,r)$ of $G$ induced
by the collection of  nodes at distance at most $r$ from $u$. The
\emph{neighborhood} $N_G(u)$ of $u$ is the set of neighbors of $u$, i.e.,
$N_G(u)= B_G(u,1)\setminus\{u\}$.
In what follows, we may omit the subscript $G$ from the previous notations when there is no risk
of confusion.\\

\paragraph{Functions:}
A function $f\colon \Reals ^\ell\to \Reals$ is \emph{non-decreasing}
if for every two vectors $\underline{x}$ and $\underline{y}$ such that
$\underline{x}$ dominates $\underline{y}$,
\[
f(\underline{y})\leq f(\underline{x}).
\]

A function $f\colon\Reals^+\to \Reals^+$ is \emph{moderately-slow} if it
is non-decreasing  and there exists a positive integer $\alpha$ such that
\[
\forall i\in\Natural\setminus\{1\},\quad\alpha\cdot f(i)\ge f(2i).
\]
In other words, $f(c\cdot i)=\bigo{f(i)}$ for every constant $c$ and every integer $i$, where the
constant hidden in the $O$ notation depends only on $c$.
An example of a moderately-slow function is given by the logarithm.

A function $f\colon\Reals^+\to \Reals^+$ is \emph{moderately-increasing} if
it is non-decreasing and there exists a positive integer $\alpha$  such that
\[
\forall i\in\Natural\setminus\{1\},\quad f(\alpha\cdot
i)\ge2f(i)\quad\text{and}\quad \alpha\cdot f(i)\ge f(2i).
\]
Note that $f(x)=x^{k_1} \log^{k_2}(x)$ is a moderately-increasing function
for every two reals $k_1\ge 1$ and $k_2\ge 0$.
Moreover, every
moderately-increasing function is moderately-slow. On the other hand, some
functions (such as the constant functions or the logarithm) are
moderately-slow but not moderately-increasing.

A function $f\colon\Reals^+\to\Reals^+$ is \emph{moderately-fast} if it is
moderately-increasing and there exists a polynomial $P$ such that
\[
\forall x\in\Reals^+,\quad x<f(x)<P(x).
\]

A function $f\colon\Reals^+\to\Reals^+$ \emph{tends to infinity} if
\[\lim\sup_{x\to\infty}f(x)\break=\infty,\] and
$f$ is \emph{ascending} if it is non-decreasing and it tends to infinity.
(Note that in this case $\lim_{x\to\infty}f(x)=\infty$.)

A function $f\colon(\Reals^+)^{\ell}\to\Reals^+$ is \emph{additive} if
there exist $\ell$ ascending functions $f_1,\ldots,f_{\ell}$ such that
\[
f(x_1,\ldots,x_\ell)=\sum_{i=1}^\ell f_i(x_i).
\]

\paragraph{Problems and instances:}
Given a set $V$ of nodes, a \emph{vector for $V$} is an assignment $\inp$ of
a bit string $\inp(v)$ to each $v\in V$, i.e., $\inp$ is a function $\inp:V\to
\{0,1\}^*$. A \emph{problem} is defined by a collection of triplets:
$\Prob=\{(G,\inp,\out)\}$, where $G$ is a (not necessarily connected)
graph, and $\inp$ and $\out$ are \emph{input} and \emph{output} vectors for
$V$, respectively. We consider only problems that are closed under disjoint
union, i.e., if $G_1$ and $G_2$ are two vertex disjoint graphs and
$(G_1,\inp_1,\out_1), (G_2,\inp_2,\out_2)\in\Prob$ then
$(G,\inp,\out)\in\Prob$, where $G= G_1\cup G_2$, $\inp= \inp_1\cup\inp_2$ and
$\out= \out_1\cup\out_2$.

 An \emph{instance}, with respect to a given
problem $\Prob$, is a pair $(G,\inp)$ for which there exists an output
vector $\out$ satisfying $(G,\inp,\out)\in \Prob$.
In what follows, whenever we consider a collection $\cF$ of instances, we always assume that
$\cF$ is closed under inclusion. That is, if $(G,\inp)\in \cF$ and $(G',\inp')\subseteq (G,\inp)$
(i.e., $G'$ is a subgraph of $G$ and $\inp'$ is the input vector $\inp$ restricted to $V(G')$) then
$(G',\inp')\in \cF$.
Informally, given a problem $\Prob$ and a collection of instances $\cF$, the goal is to design an efficient
distributed algorithm that takes an instance  $(G,\inp)\in\cF $ as input, and
produces an output vector $\out$ satisfying $(G,\inp,\out)\in \Prob$. The
reason to require problems to be closed under disjoint union is that a
distributed algorithm operating on an instance $(G,\inp)$ runs separately and
independently on each connected component of $G$.
Let $\cal{G}$ be a family of graphs closed under inclusion.
We define $\cF(\cal{G})$ to be $\{G\}\times\{0,1\}^*$.

We assume that each node
$v\in V$ is  provided with  a unique integer referred to as the
\emph{identity} of $v$, and denoted $\id(v)$; 
by unique identities, we mean
that $\id(u)\neq \id(v)$ for every two distinct nodes $u$ and $v$.
For ease of exposition, we consider the identity of a node to be
part of its input.

We consider classical problems such as  coloring, maximal matching ($\mm$), Maximal Independent Set ($\MIS$) and the $(\alpha,\beta)$-ruling
set problem.
Informally,  viewing the output of a node as a \emph{color},
the requirement of \emph{coloring} is that the colors of two neighboring nodes must be different.
In the $(\alpha,\beta)$-ruling set problem, the output at each node is Boolean,
and indicates whether
the node belongs to a set $S$ that must form an
$(\alpha,\beta)$-ruling set. That is,  the set $S$ of selected nodes must satisfy:  (1) two nodes that belong to $S$ must be at distance at least $\alpha$ from each other, and (2) if a node does not belong to $S$, then there is a node in the set at distance at most $\beta$ from it. $\MIS$ is a special case of the ruling set problem, specifically, $\MIS$ is  precisely $(2,1)$-ruling set.
Finally, given a triplet
$(G,\inp,\out)$, two nodes $u$ and $v$ are said to be \emph{matched} if $(u,v)\in E$,
$\out(u)=\out(v)$ and  $\out(w)\neq\out(u)$ for every $w\in(N_G(u)\cup
N_G(v))\setminus\{u,v\}$.
Thus, the $\mm$ problem  requires that each node $u$ is either matched to one of its neighbors or that every neighbor $v$  of $u$ is matched to one of $v$'s neighbors.\\

\paragraph{Parameters:} Fix a problem $\Prob$ and let $\cF$ be a collection
of instances for $\Prob$.
A \emph{parameter} $\parameter$ is a positive valued function
$\para:\cF\to \Natural$. The parameter $\parameter$ is
\emph{non-decreasing}, if $\para(G',\inp')\leq \para(G,\inp)$ whenever $(G',\inp')\in \cF$
and $(G',\inp')\subseteq (G,\inp)$.

Let $\cF$ be a collection of instances.
A parameter $\para$ for $\cF$ is a \emph{graph-parameter} if $\para$ is
independent of the input, that is, if $\para(G,\inp)=\para(G,\inp')$ for every two instances $(G,\inp),(G,\inp')\in\cF$
such that the input  assignments $\inp$ and $\inp'$ preserve the identities, i.e., the inputs $\inp(v)$ and $\inp'(v)$ contain the same identity $\id(v)$ for every $v\in V(G)$.
In what follows, we will consider only non-decreasing graph-parameters
(note, not all graph-parameters are non-decreasing, an example being the
diameter of a graph). More precisely,
we will primarily focus on the following non-decreasing graph-parameters:
the number $n$ of nodes of the graph $G$,
i.e.,~$\abs{V(G)}$,
the maximum degree $\Delta=\Delta(G)$ of $G$, i.e.,~$\max\sst{\deg_G(u)}{u\in V(G)}$,  and
the arboricity $a=a(G)$ of $G$, i.e.,~the least number of acyclic subgraphs
of $G$ whose union is $G$.\\

\paragraph{Local algorithms:} Consider a problem $\Prob$ and a collection of
instances $\cF$ for $\Prob$.  An algorithm for $\Prob$ and $\cF$ takes as
input an instance $(G,\inp)\in\cF$ and must terminate with an output vector
$\out$  such that $(G,\inp,\out)\in \Prob$.  We consider the $\cal{LOCAL}$
model (cf.,~\cite{Pel00}).
During the execution of a \emph{local} algorithm~$\cA$, all processors are
woken up simultaneously and computation proceeds in fault-free synchronous
rounds. In each round, every node may
send messages of unrestricted size to its neighbors and may perform
arbitrary computations on its data. A message that is sent in a round $r$,
arrives at its destination before the next round $r+ 1$~starts.  It must be
guaranteed that after a finite number of rounds, each node $v$ terminates
by writing  some final output value~$\out(v)$ in its designated output variable (informally, this means that we may assume that  a node ``knows'' that its
output is indeed its final output.) The algorithm $\cA$ is \emph{correct}
if for every instance $(G,\inp)\in \cF$, the resulting output vector $\out$
satisfies $(G,\inp,\out)\in \Prob$.

Let $\cA$ be a local deterministic  algorithm for $\Prob$ and $\cF$.  The
\emph{running time} of $\cA$ over a particular instance $(G,\inp)\in\cF$,
denoted $T_{\cA}(G,\inp)$, is the number of rounds from the beginning of the
execution of $\cA$ until all nodes terminate. The running time of $\cA$
is typically evaluated with respect to a collection $\Lambda$ of  parameters
$\q_1,\ldots,\q_\ell$.  Specifically, it is compared to a non-decreasing
function $f\colon \Natural ^\ell\to \Reals^+$; we say that $f$ is  an
upper bound for the running time of $\cA$ with respect to $\Lambda$ if
$T_{\cA}(G,\inp)\leq f(\q^*_1,\ldots,\q^*_\ell)$ for every instance
$(G,\inp)\in F$ with parameters
$\q^*_i=\q_i(G,\inp)$ for $i\in\icc{1}{\ell}$. Let us stress that we
assume throughout the paper that all the functions bounding running times of algorithms are
non-decreasing.

For an integer $i$, the algorithm $\cA$ \emph{restricted to $i$ rounds}
is the local algorithm $\cB$ that consists of running $\cA$ for precisely
$i$ rounds. The output $\out(u)$ of $\cB$ at a vertex $u$ is defined as
follows: if, during the $i$ rounds, $\cA$ outputs a value $y$ at $u$ then
$\out(u)= y$; otherwise we let $\out(u)$ be an arbitrary value,
e.g.,~``$0$''.\\

A \emph{randomized} local algorithm is a local algorithm that allows each node to use
random bits in its local computation|the random bits used by different
nodes being independent.
A randomized (local) algorithm $\cA$ is \emph{Las Vegas} if
its correctness is guaranteed with probability $1$.
The \emph{running time} of
a Las Vegas algorithm $\cA_{LV}$ over a particular configuration
$(G,\inp)\in\cF$, denoted $T_{\cA_{LV}}(G,\inp)$, is a random variable, which
may be unbounded. However, the expected value of $T_{\cA_{LV}}(G,\inp)$ is
bounded.
A \emph{Monte-Carlo} algorithm $\cA_{MC}$
with guarantee $\rho\in\ioc{0}{1}$ is a randomized algorithm that takes a configuration $(G,\inp)\in\cF$ as input and  terminates before
a predetermined time $T_{\cA_{MC}}(G,\inp)$ (called the \emph{running time} of $\cA_{MC}$). It is certain that the output vector produced
by Algorithm $\cA_{MC}$ is a solution to $\Prob$ with probability at least $\rho$.
Finally, a \emph{weak Monte-Carlo} algorithm $\cA_{WMC}$ with guarantee
$\rho\in\ioc{0}{1}$ guarantees that with probability at least $\rho$,
the algorithm outputs a correct solution by its running time
$T_{\cA_{WMC}}(G,\inp)$.
(Observe that
it is not certain that any execution of the weak Monte-Carlo algorithm will terminate by the
prescribed time $T_{\cA_{WMC}}(G,\inp)$, or even terminate at all.)
Note that a Monte-Carlo algorithm is in particular a weak Monte-Carlo algorithm, with the same
running time and guarantee. Moreover, for any constant
$\rho\in\ioc{0}{1}$, a Las Vegas algorithm running in expected time $T$ is a
weak Monte-Carlo algorithm with guarantee $\rho$ running in time
$\frac{T}{1-\rho}$, by Markov's inequality.

\paragraph{Synchronicity and time complexity:}
Many \LOCAL{} algorithms happen to have different
termination times at different nodes. On the other hand, most of
the algorithms rely on a simultaneous wake-up time for all nodes. This
becomes an issue when one wants to run an algorithm $\cA_1$ and subsequently
an algorithm $\cA_2$ taking the output of $\cA_1$ as input.
Indeed, this problem amounts to running $\cA_2$ with
non-simultaneous wake-up times: a node $u$ starts $\cA_2$ when it
terminates $\cA_1$.

As observed (e.g., by Kuhn~\cite{Kuh09}), the concept of
synchronizer~\cite{Awe85}, used in the context of local algorithms, allows one to
transform an asynchronous  local algorithm to a synchronous one that runs in the same asymptotic time complexity.
Hence,  the synchronicity assumption can actually be removed.
Although the  standard asynchronous model introduced still assumes
a simultaneous wake-up time, it can be easily verified that the technique still applies with
non-simul\-ta\-neous wake-up times if a node can buffer messages received
before it wakes up, which is the case when running an algorithm after
another.

However, we have to adapt the notion of running time.
The computation that a node performs in time $t$ depends on its interactions
with nodes at distance at most $t$ in the network. 
More precisely, we say that a node $u$
terminates in time $t$ if it terminates at most $t$ rounds after all
nodes in $B_G(u,t)$ have woken up. The termination time of $u$ is the
least $t$ such that $u$ terminates in time $t$.
We finally define the running time of an algorithm as the maximum
termination time over all nodes and all wake-up patterns.

Given two local algorithms $\cA_1$ and $\cA_2$, we let
$\cA_1;\cA_2$ be the process of running $\cA_2$ after $\cA_1$.
It turns out that the running time of $\cA_1;\cA_2$ is bounded
from above by the sum of the running times of $\cA_1$ and $\cA_2$.
This can be shown as follows.
Let $t_1$ and $t_2$ be the running times of $\cA_1$ and $\cA_2$
respectively.
Consider a node $u$ and let $t_0$ be the last wake-up time of a node in
the ball $B_G(u,t_1+t_2)$.
At $t_0+t_1$, all nodes in $B_G(u,t_2)$ have terminated $\cA_1$ and are
thus considered as woken up
for the execution of $\cA_2$. Node $u$ thus terminates before
$(t_0+t_1)+t_2$. As this is true for any node $u$
independently of the wake-up pattern, $\cA_1;\cA_2$ has running time at
most $t_1+t_2$. This establishes the following observation.

\begin{observ}\label{obs:runafter}
For any two local algorithms $\cA_1$ and $\cA_2$,
the running time of $\cA_1;\cA_2$
is bounded by the sum of the running times of $\cA_1$ and $\cA_2$.
\end{observ}

Another useful remark is that
a simultaneous wake-up algorithm running in time $t$ can be
emulated in a non-simultaneous wake-up environment with running time at most $t$ using the simple $\alpha$
synchronizer. Indeed, consider a node $u$ and let
$t_0$ be the last wake-up time of a node in the ball $B_G(u,t)$.
At time $t_0$, all nodes in $B_G(u,t)$ perform (or have performed) round $0$.
Using the $\alpha$ synchronizer a node can perform round $i$ when
all its neighbors have performed round $i-1$. We can thus show by
induction on $i$ that all nodes in $B_G(u,t-i)$ perform
(or have performed) round $i$ at time $t_0+i$. The node $u$ thus
terminates in time $t$.
This implies that the running time of the
emulation of the algorithm with the $\alpha$ synchronizer is at most $t$.
Therefore, in the remaining of the paper we may assume without loss of
generality that all nodes wake up simultaneously at time $0$.

\paragraph{Local algorithms requiring parameters:}
Fix a problem $\Prob$ and let $\cF$ be a collection of instances for
$\Prob$.
Let $\Gamma$ be a collection of
parameters $\para_1,\ldots,\para_r$ and let $\cA$ be a
local algorithm.
We say that $\cA$ \emph{requires} $\Gamma$ if the code of $\cA$,
which is executed by each node of the input configuration, uses a value
$\tilde{\para}$ for each parameter $\para\in \Gamma$. (Note that this value is
thus the same for all nodes.)
The value $\tilde{\para}$ is a \emph{guess} for~$\para$.
A collection of guesses for the parameters in $\Gamma$ is denoted by
$\tilde{\Gamma}$ and an algorithm $\cA$ that requires $\Gamma$ is denoted by $\cA^\Gamma$.
An algorithm that does not require any parameter~is~called~\emph{uniform}.

Consider an instance $(G,\inp)\in\cF$, a collection $\Gamma$ of parameters
and a parameter $\para\in\Gamma$. A guess $\tp$ for $\para$ is termed \emph{good}
if $\tp\geq \para(G,\inp)$, and the guess $\tp$ is called \emph{correct} if $\tp= \para(G,\inp)$.
We typically write correct guesses and collection of correct guesses with a
star superscript, as in $\para^*$ and ${\Gamma}^*(G,\inp)$, respectively. When $(G,\inp)$ is clear from the context, we may use the notation ${\Gamma}^*$ instead of ${\Gamma}^*(G,\inp)$.

An algorithm  $\cA^\Gamma$ \emph{depends} on $\Gamma$ if for every instance
$(G,\inp)$ $\in \cF$, the correctness of $\cA^\Gamma$ over $(G,\inp)$ is
guaranteed only when $\cA^\Gamma$ uses a collection $\tilde{\Gamma}$ of good
guesses.

Consider an algorithm $\cA^\Gamma$ that depends on a collection $\Gamma$ of
parameters $\para_1,\ldots,\para_r$ and fix an instance $(G,\inp)$.
Observe that the  running time of $\cA^\Gamma$ over $(G,\inp)$ may be different for different
collections of guesses $\tilde{\Gamma}$, in other words, the  running time over $(G,\inp)$
may be a function of $\tilde{\Gamma}$. Recall that when we consider
an algorithm that does not require parameters, we still typically evaluate its running time
with respect to a collection of parameters $\Lambda$. We generalize this to the case
where the algorithm depends on $\Gamma$ as follows.

Consider two collections $\Gamma$ and $\Lambda$ of parameters
$\para_1,\ldots, \para_r$ and $\q_1,\ldots,\q_{\ell}$, respectively.
Some parameters may belong to both $\Gamma$ and $\Lambda$.
Without loss of generality, we shall always assume that
$\{\para_{r'+1},\ldots, \para_r\}\cap
\{\q_{r'+1},\ldots,\q_\ell\}=\emptyset$ for some
$r'\in\icc{0}{\min\{r,\ell\}}$ and $\para_i=\q_i$ for every
$i\in\icc{1}{r'}$.
Notice that $\Gamma\setminus \Lambda=\{\para_{r'+1},\para_{r'+2},\ldots,\para_r\}$.
A function $f\colon(\Reals^+)^{\ell}\to\Reals^+$
\emph{upper bounds} the  running time
of $\cA^{\Gamma}$ with respect to $\Gamma$ and  $\Lambda$ if the
 running time  $T_{\cA^{\Gamma}}(G,\inp)$ of $\cA^{\Gamma}$  for $(G,\inp)\in\cF$ using a collection of good guesses $\tilde{\Gamma}=\{\tp_1,\ldots,\tp_r\}$
is at most $f(\tp_1,\ldots,\tp_{r'},\ldots,\q^*_{\ell})$, where
$\q^*_i=\q_i(G,\inp)$ for $i\in\icc{r'+1}{\ell}$. Note that we do not put any restriction
on the  running time of $\cA^{\Gamma}$ over $(G,\inp)$ if some of the guesses in $\tilde{\Gamma}$
are not good. In fact, in such a case, the algorithm may not even terminate
and it may also produce wrong results.

For simplicity of notation, when $\Gamma$ and $\Lambda$ are clear
from the context, we
say that $f$ upper bounds the  running time of $\cA^{\Gamma}$, without
writing that it is with respect to $\Gamma$ and  $\Lambda$.

The set $\Gamma$ is \emph{weakly-dominated} by $\Lambda$ if for each
$j\in\icc{r'+1}{r}$, there exists an index $i_j\in\icc{1}{\ell}$ and
an ascending function $g_j$ such that $g_j(\para_j(G,\inp))\leq \q_{i_j}(G,\inp)$
for every instance $(G,\inp)\in\cF$.
(For example, $\Gamma=\{\Delta\}$ is weakly-dominated by  $\Lambda=\{n\}$,
since $\Delta(G,\inp)\leq n(G,\inp)$ for any $(G,\inp)$.)

\section{Pruning Algorithms}\label{sec:pruning}

\subsection{Overview}

Consider  a problem $\Prob$ in the centralized setting and an efficient randomized
Monte-Carlo algorithm $\cA$ for $\Prob$. A known method for
transforming $\cA$ into a Las Vegas algorithm is based on repeatedly doing the following.
Execute $\cA$ and, subsequently, execute an algorithm that
checks the validity of the output.  If the checking fails then continue,
and otherwise, terminate, i.e., break the loop.  This transformation can
yield  a Las Vegas algorithm whose expected running time is  similar
to the running time of the Monte-Carlo algorithm provided that
the checking mechanism used is efficient.

If we wish to come up with a similar transformation in the context of
locality, a first idea would be to consider a local algorithm that checks
the validity of a tentative output vector. This concept has been studied
from various perspectives (cf., e.g.,~\cite{FKP11,KKP10,NaSt95}).  However,
such fast local checking~procedures can only guarantee that faults are
detected by at least one node, whereas to restart the Monte-Carlo algorithm,
all nodes should be aware of a fault.  This notification~can~take diameter
time and will thus violate the locality constraint (i.e. running in a bounded
number of rounds).

Instead of using local checking procedures, we introduce the notion of
\emph{pruning algorithms}. Informally, this is a  mechanism that identifies
``valid areas'' where the tentative output vector~$\hout$ is valid and \emph{prunes} these
areas, i.e.,  takes them out of further consideration. A pruning
algorithm $\cP$  must satisfy two properties, specifically, (1)  \emph{gluing}:  $\cP$ must make sure that the current solution
on these ``pruned areas'' can be extended to a valid solution for the remainder of
the graph, and (2)  \emph{solution detection}:  if  $\hout$ is a valid global solution to begin with then $\cP$ should prune all nodes.
Observe that since the empty output
  vector is  a solution for the empty input graph then (1) implies the converse of (2), that is, if
$\cP$ prunes all nodes, then $\hout$ is a valid global solution.

Now, given a Monte-Carlo algorithm $\cA$ and a pruning algorithm $\cP$ for
the problem, we can transform $\cA$ into a Las Vegas algorithm by
executing the pair of algorithms $(\cA;\cP)$ in iterations,
where each iteration $i$ is executed on the graph $G_i$ induced by the set
of  nodes that were not pruned in previous iterations ($G_1$ is the initial
graph $G$).  If, in some iteration $i$, Algorithm $\cA$ solves the problem
on the graph $G_i$, then the solution detection property guarantees that the
subsequent pruning algorithm will prune all nodes in $G_i$ and hence at that
time all nodes are pruned and the execution terminates. Furthermore, using
induction, it can be shown that the gluing property guarantees that the
correct solution to $G_i$ combined with the outputs of the previously pruned
nodes forms a solution to $G$.

\subsection{Pruning Algorithms: Definition and Examples.}
We now formally define pruning algorithms.
Fix a problem  $\Prob$ and a family of instances $\cF$ for $\Prob$.  A
\emph{pruning} algorithm $\cP$ for $\Prob$ and $\cF$ is a uniform algorithm
that takes as input a triplet $(G,\inp,\hout)$, where $(G,\inp)\in\cF$ and
$\hout$ is some tentative output vector (i.e. an output vector that may be incorrect), and returns a configuration
$(G',\inp')$ such that $G'$ is an induced subgraph of $G$ and $(G',\inp')\in\cF$.
Thus, at each node $v$ of $G$, the pruning algorithm $\cP$ returns a bit $b(v)$
that indicates whether $v$ belongs to some selected subset $W$ of nodes of $G$ to be pruned.
(Recall that the idea is to assume that nodes in $W$ have a satisfying tentative output
value and that they can be excluded from further computations.)
Note that $\inp'$ may be different than $\inp$ restricted to the nodes outside
$W$.

Consider now an output vector $\out'$ for the nodes in $V(G')$.  The
\emph{combined} output vector $\out$ of the vectors $\hout$  and
$\out'$  is the output vector that is a combination of  $\hout$ restricted
to the nodes in $W$ and $\out'$ restricted to  the nodes in~$G'$,
i.e.,~$\out(v)=\hout(v)$ if $v\in W$ and $\out(v)=\out'(v)$
otherwise.  A pruning algorithm $\cP$ for a problem $\Prob$ must satisfy the following~properties.

\begin{itemize}
\item \textbf{Solution detection:} if $(G,\inp,\hout)\in\Prob$, then
$W=V(G)$, that is,
$\cP(G,\inp,\hout)=(\emptyset,\emptyset)$.\\
\item \textbf{Gluing:} if  $\cP(G,\inp,\hout)=(G',\inp')$ and $\out'$ is a solution for
  $(G',\inp')$, i.e., $(G',\inp',\out')\in \Prob$, then the combined output vector $\out$
is a solution for $(G,\inp)$, i.e, $(G,\inp,\out)\in \Prob$.
\end{itemize}

As mentioned earlier, it follows from the gluing property that  if the pruning
algorithm $\cP$ returns $(\emptyset,\emptyset)$
(i.e., all nodes are pruned) then $(G,\inp,\hout)\in \Prob$.

The pruning algorithm $\cP$ is \emph{monotone with respect to a parameter}
$\para$ if $\parameter(G,\inp)\ge \para(\cP(G,\inp,\hout))$ for every $(G,\inp)\in\cF$
and every tentative output vector $\hout$.
The pruning algorithm $\cP$ is \emph{monotone with respect to a
collection of parameters} $\Gamma$ if $\cP$ is monotone with respect to
every parameter $\para\in \Gamma$. In such a case, we may also say that
$\cP$ is \emph{$\Gamma$-monotone}.
The following assertions follow from the definitions.

\begin{observ}\label{claim:mono} Let $\cP$ be a pruning algorithm.
\begin{enumerate}
\item Algorithm $\cP$ is monotone with respect to any non-decreasing
graph-parameter.
\item If the configuration $(G',\inp')$ returned by $\cP$ satisfies
$\inp'(v)=\inp(v)$ for every $v\in V(G)\setminus W$ and every
configuration $(G,\inp)$, then $\cP$ is monotone
with respect to any non-decreasing parameter.
\end{enumerate}
\end{observ}

For simplicity,  we impose that the running time of a pruning algorithm $\cP$
be constant.
We shall elaborate on general pruning algorithms at the end of the paper.

We now give examples of pruning algorithms for several problems,
namely, $(2,\beta)$-Ruling set for a constant integer $\beta$ (recall that  MIS is precisely  $(2,1)$-Ruling set), and maximal matching.
These pruning algorithms ignore the input of the nodes.
Thus, by Observation~\ref{claim:mono}, they
are monotone with respect to any non-decreasing parameter.

\paragraph{The $(2,\beta)$-ruling set pruning algorithm:} Let $\beta$ be a constant integer. We define a pruning algorithm $\cP_{(2,\beta)}$ for the  $(2,\beta)$-ruling set problem as follows. Given a triplet $(G,\inp,\hout)$,
let $W$ be the set of nodes $u$ satisfying one of the following two conditions.
\begin{itemize}
\item
$\hout(u)=1$ and $\hout(v)=0$  for all  $v \in N(u)$, or
\item
$\hout(u)=0$ and $\exists v\in B_G(u,\beta)$ such that
$\hout(v)=1$ and  $\hout(w)=0$  for all  $w \in N(v)$.
\end{itemize}

The question of whether a node $u$ belongs to $W$ can be determined by
inspecting $B_G(u,1+\beta)$, the ball of radius $1+\beta$ around $u$. Hence, we obtain the following.\\

\begin{observ}\label{claim:universal-ruling}
Algorithm $\cP_{(2,\beta)}$  is a pruning algorithm for the $(2,\beta)$-ruling set
problem, running in time $1+\beta$. (In particular, $\cP_{(2,1)}$  is a pruning algorithm for the MIS problem running in time 2.)
Furthermore, $\cP_{(2,\beta)}$ is monotone with
respect to any non-decreasing parameter.
\end{observ}

\paragraph{The maximal matching problem:}
We define a pruning algorithm $\cP_{\mm}$ as follows.
Given a tentative output vector $\hout$, recall that $u$
and $v$ are matched when $u$ and $v$ are neighbors, $\hout(u)=\hout(v)$ and
$\hout(w)\neq\hout(u)$ for every $w\in(N_G(u)\cup N_G(v))\setminus\{u,v\}$.
Set $W$ to be the set of nodes $u$ satisfying one of  the following conditions.
\begin{itemize}
\item
$\exists v\in N(u)$ such that $u$ and $v$ are matched, or
\item
$\forall v\in N(u)$, $\exists w\neq u$ such that $v$ and $w$ are matched.\\
\end{itemize}

\begin{observ}\label{claim:universal-mm}
Algorithm $\cP_{\mm}$ is a  pruning algorithm for $\mm$ whose
running time is 3. Furthermore, $\cP_{\mm}$ is monotone with respect to any parameter.
\end{observ}

We exhibit several applications of pruning algorithms.
The main application  appears in the next
section, where we show how pruning algorithms can be used to transform
non-uniform algorithms into uniform ones.
Before we continue, we need the concept of alternating algorithms.

\subsection{Alternating Algorithms}
A pruning algorithm can be used in conjunction with a sequence of
algorithms as follows.
Let $\cF$ be a collection of
instances for some problem $\Prob$. For each $i\in\mathbf{N}$,
let $\cA_i$ be an algorithm defined on $\cF$.
Algorithm $\cA_i$ does not necessarily
solve $\Prob$, it is only assumed to produce some output.

Let $\cP$ be a pruning algorithm for $\Prob$ and $\cF$, and for $i\in\mathbf{N}$, let
$\cB_i= (\cA_i;\cP)$, that is, given an
instance $(G,\inp)$, Algorithm $\cB_i$ first executes
$\cA_i$, which returns an output vector $\out$
for the nodes of $G$ and,
subsequently, Algorithm $\cP$ is executed over the triplet
$(G,\inp,\out)$.
We define the \emph{alternating algorithm} $\pi$ for
$(\cA_i)_{i\in\mathbf{N}}$ and $\cP$ as follows.  The
alternating algorithm $\pi=\pi((\cA_i)_{i\in\mathbf{N}},\cP)$ executes the algorithms
$\cB_i$ for $i=1,2,3,\ldots$ one after the other: let $(G_1,\inp_1)=(G,\inp)$
be the initial instance given to $\pi$;
for $i\in\mathbf{N}$, Algorithm $\cA_i$ is executed on the
instance $(G_{i},\inp_{i})$ and returns the output vector~$\out_{i}$.
The subsequent pruning algorithm $\cP$ takes the triplet $(G_{i},\inp_{i},\out_{i})$ as
input
and produces the instance $(G_{i+1},\inp_{i+1})$. See Figure~\ref{fig:alt}
for a schematic view of an alternating algorithm. The definition extends to
a finite sequence $(\cA_i)_{i=1}^{k}$ of algorithms in
a natural way; the alternating algorithm for $(\cA)_{i=1}^{k}$ and $\cP$
being $A_1;\cP;A_2;\cP;\cdots;A_k;\cP$.

\begin{figure*}[!ht]
  \begin{center}
    \begin{tikzpicture}[scale=.95]
      \node (1) at (0,0) {$(G_1,\inp_1)$};
      \node (v2) at (2.6*\unit,0) {};
  \draw[->,>=latex,dashed] (1)--(v2) node[midway,above] {$\cB_1=\cA_1;\cP$}
    node[midway,below=1.5*\unit] (I1) {$(G_1,\inp_1,\out_1)$}
    node[right] (2) {$(G_2,\inp_2)$};
  \draw[->,>=latex,dashed] (2)--++(2.6*\unit,0) node[midway,above] {$\cB_2=\cA_2;\cP$}
    node[midway,below=1.5*\unit] (I2) {$(G_2,\inp_2,\out_2)$}
    node[right] (3) {$(G_3,\inp_3)$};
  \draw[->,>=latex,dotted] (3)--++(1.8*\unit,0)
  node[right] (4) {$(G_i,\inp_i)$};
  \draw[->,>=latex,dashed] (4)--++(2.6*\unit,0) node[midway,above] {$\cB_i=\cA_i;\cP$}
    node[midway,below=1.5*\unit] (I3) {$(G_i,\inp_i,\out_i)$}
    node[right] (5) {$(G_{i+1},\inp_{i+1})$};
  \draw[->,>=latex,dotted] (5)--++(1.8*\unit,0);
        \draw[->,>=latex] (1)--(I1) node[sloped,midway,above] {$\cA_1$};
        \draw[->,>=latex] (I1)--(2) node[sloped,midway,above] {$\cP$};
        \draw[->,>=latex] (2)--(I2) node[sloped,midway,above] {$\cA_2$};
        \draw[->,>=latex] (I2)--(3) node[sloped,midway,above] {$\cP$};
        \draw[->,>=latex] (4)--(I3) node[sloped,midway,above] {$\cA_{i}$};
        \draw[->,>=latex] (I3)--(5) node[sloped,midway,above] {$\cP$};
    \end{tikzpicture}
    \caption{Schematic view of an alternating algorithm for
    $(\cA_i)_{i\in\mathbf{N}}$ and
    $\cP$.}\label{fig:alt}
  \end{center}
\end{figure*}

The alternating algorithm $\pi$ \emph{terminates} on an instance
$(G,\inp)\in\cF$ if there exists $k$ such that  $V(G_k)=\emptyset$.  Observe
that in such a case, the tail $\cB_{k};
\cB_{k+1};\cdots$ of $\pi$ is trivial. The output
vector $\out$ of a terminating alternating algorithm $\pi$ is defined
as the combination of the output vectors $ \out_{1},
\out_{2},\out_{3},\ldots$.  Specifically, for $s\in\icc{1}{k-1}$, let
$W_s= V(G_{s})\setminus V(G_{s+1})$.
(Observe that $W_s$ is precisely the set of nodes pruned by the execution
of the pruning algorithm $\cP$ in $\cB_{s}$.)
Then, the collection $\sst{W_s}{1\leq s \leq
k-1}$ forms a partition of $V(G)$, i.e., $W_s\cap W_{s'}=\emptyset$ if
$s\neq s'$, and $\cup_{s=1}^{k-1} W_s=V(G)$. Observe that the final output
$\out$ of $\pi$ satisfies $\out(u)=\out_{s}(u)$ for every node $u$,
where $s$ is such that $u\in W_s$.
In other words, the output of $\pi$
restricted to the nodes in $W_s$ is precisely the corresponding output of
Algorithm $\cA_{s}$. The next observation readily follows from the definition
of pruning algorithms.

\begin{observ}\label{sequence}
Consider a problem $\Prob$, a collection of instances $\cF$, a sequence  of algorithms $(\cA_i)_{i\in\mathbf{N}}$ defined on $\cF$ and
a pruning algorithm $\cP$ for $\Prob$ and $\cF$. Consider
the alternating algorithm $\pi=\pi((\cA_i)_{i\in\mathbf{N}},\cP)$
for $(\cA_i)_{i\in\mathbf{N}}$ and $\cP$.
If $\pi$ terminates on an instance  $(G,\inp)\in\cF$ then it produces a correct output $\out$, that is, $(G,\inp,\out)\in\Prob$.
\end{observ}

In what follows, we often produce a sequence of algorithms
$(\cA_i)_{i\in\mathbf{N}}$ from an algorithm $\cA^\Gamma$ requiring
a collection $\Gamma$ of non-decreasing parameters. The general idea is to design a
sequence of guesses $\tilde\Gamma_i$
and let $\cA_i$ be algorithm
$\cA^\Gamma$ provided with guesses $\tilde\Gamma_i$.
Given a pruning algorithm $\cP$, we obtain a uniform
alternating algorithm $\pi=\pi((\cA_i)_{i\in\mathbf{N}},\cP)$.
The  sequence of guesses is designed  such that
for any configuration  $(G,\inp)\in\cF$, there exists some $i$ for which
$\tilde\Gamma_i$ is a collection of good guesses for
$(G,\inp)$. The crux is to obtain an execution time
for $\cA_1;\cP;\cdots;\cA_i;\cP$ of the same order
as the execution time of $\cA^\Gamma$ provided with the collection
$\Gamma^*(G,\inp)$ of correct guesses.

\section{The General Method}
We now turn to the main application of pruning algorithms discussed in this
paper, that is, the construction of a transformer taking a non-uniform
algorithm $A^\Gamma$ as a black box and producing a uniform one that enjoys
the same (asymptotic) time complexity as the original non-uniform algorithm.

We begin with a few illustrative examples of our method in
Subsection~\ref{sub:illus}. Then, the general framework of our transformer
is given in Subsection~\ref{sub:genframe}. This subsection introduces
a concept of ``sequence-number functions'' as well as the
a fundamental construction used in our forthcoming algorithms.

Then, in Subsection~\ref{sub:deterministic}, we consider  the deterministic
setting: a somewhat restrictive, yet useful, transformer is given in
Theorem~\ref{simple}.  This transformer considers a single set $\Gamma$ of
non-decreasing parameters $\para_1,\ldots,\para_{\ell}$, and assumes
that (1) the given non-uniform algorithm $\cA^{\Gamma}$ depends on ${\Gamma}$
and (2) the running time of $\cA^{\Gamma}$ is evaluated with respect to the parameters in
$\Gamma$. Such a situation is customary, and occurs for instance for the
best currently known MIS Algorithms~\cite{BaEl09,Kuh09,PaSr96}
as well as for the maximal matching algorithm of Hanckowiak \emph{et
al.}~\cite{HKP01}.
As a result, the transformer given by Theorem~\ref{simple} can be used
to transform each of these algorithms into a uniform
one with asymptotically the same time complexity.

The transformer of Theorem~\ref{simple} is extended to the randomized
setting in Subsection~\ref{sub:randomized}.
In Subsection~\ref{sub:general}, we establish Theorem~\ref{thm:parameters}, which
generalizes both Theorem~\ref{simple} and Theorem~\ref{thm:randomized}. Finally, we
conclude the section with Theorem~\ref{thm:min} in
Subsection~\ref{sec:min}, which shows how to manipulate several uniform
algorithms that run in unknown times to obtain a uniform algorithm that runs
as fast as the fastest algorithm among those given algorithms.

\subsection{Some Illustrative Examples}\label{sub:illus}
The basic idea is very simple. Consider a problem for which~we have a
pruning algorithm $\cP$, and a non uniform~algorithm~$\cA$ that requires the
upper bounds on some parameters to be part of the input. To obtain a uniform
algorithm, we execute the pair of algorithms $(\cA;\cP)$ in iterations,
where each iteration  executes $\cA$ using a specific set of guesses for the
parameters. Typically, as iterations proceed, the guesses for the parameters
grow larger and larger until we reach  an iteration $i$ where all the
guesses are larger than the actual value of the corresponding parameters. In
this iteration, the operation of $\cA$ on~$G_i$
 using such guesses
guarantees a correct solution on $G_i$ ($G_i$ is the graph induced by the set
of  nodes that were not pruned in previous iterations). The solution detection property of the pruning algorithm
then guarantees that the execution terminates in this
iteration  and hence, Observation~\ref{sequence} guarantees that the output of all nodes combines to a
global solution on $G$.  To bound the running time, we shall make sure that
the total running time is dominated by the running time of the last
iteration, and that this last iteration is relatively fast.

There are various delicate points when using this general strategy. For
example, in iterations where incorrect guesses are used, we have no control
over the behavior of the non-uniform algorithm $\cA$ and, in particular, it
may run for too many rounds, perhaps even indefinitely.  To overcome this
obstacle, we allocate a prescribed  number of rounds for each iteration;  if
Algorithm $\cA$  reaches this time bound without outputting at some node
$u$, then we force it to terminate with an arbitrary output. Subsequently,
we run the pruning algorithm and proceed to the next iteration.

Obviously,  this simple approach of running in iterations and increasing the
guesses from iteration to iteration is hard\-ly new.  It was used, for
example,
in the context of wireless networks to compute estimates of parameters (cf.,
e.g.,~\cite{BeYa76,NaOl02}), or to estimate the number of faults~\cite{KuPe00}.
It was also used by Barenboim and Elkin~\cite{BaEl10b} to avoid the necessity
of having an upper bound on the arboricity $a$ in one of their MIS algorithms,
although their approach increases the running time by $\log^* n$.
One of the main contributions of the current
paper is the formalization and generalization of this technique, allowing it
to be used for a wide varieties of problems and applications. Interestingly,
note that we are only concerned with getting rid of the use of some global
parameters in the code of local algorithms,
and not with obtaining estimates for them (in particular, when
our algorithms terminate, a node has no guarantee to have
upper bounds on these global parameters).

To illustrate the method, let us consider the non-uniform MIS algorithm of
Panconesi and Srinivasan~\cite{PaSr96}. The code of Algorithm $\cA$ uses
an upper bound $\tilde{n}$ on the number of nodes $n$, and runs
in time at most $f(\tilde{n})=2^{\bigo{\sqrt{\log \tilde{n}}}}$. Consider a
pruning algorithm $\cP_{\MIS}$ for MIS (such an algorithm is given by
Observation \ref{claim:universal-ruling}).  The following sketches our
technique for obtaining a uniform MIS algorithm.  For each integer $i$,
set $n_i=\max\sst{a\in\mathbf{N}}{f(a)\le2^i}$.
In Iteration $i$,
for $i=1,2,\ldots$, we first execute Algorithm $\cA$ using the guess $n_i$
(as an input serving as an upper bound for the number of nodes) for
precisely $2^i$ rounds.
Subsequently, we run the pruning algorithm $\cP_{\MIS}$.  When the pruning
algorithm terminates, we execute the next iteration on the non-pruned nodes.
Let $s$ be the integer such that $2^{s-1}<f(n)\leq 2^s$, where $n$ is the
number of nodes of the input graph. By the definition, $n\leq n_s$.
Therefore, the application of $\cA$ in Iteration $s$ uses a guess $n_s$ that
is indeed good, i.e.,~larger than the number of nodes.  Moreover, this
execution of $\cA$
is completed before the prescribed deadline of $2^{s}$ rounds expires
because its running time is  at most $f(n_s)\leq 2^{s}$. Hence, we are
guaranteed to have a correct solution by the end of Iteration~$s$. The
running time is thus at most $\sum_{i=1}^s 2^i=\bigo{f(n)}$.

This method can sometimes be extended to simultaneously remove the use
of several parameters in the code of a local algorithm.  For example, consider the MIS
algorithm of Barenboim and Elkin~\cite{BaEl09} (or that of
Kuhn~\cite{Kuh09}), which uses upper bounds $\tilde{n}$
and $\tilde{\Delta}$ on $n$ and $\Delta$, respectively,  and runs in time
$f(\tilde{n},\tilde{\Delta})=f_1(\tilde{n})+f_2(\tilde{\Delta})$, where
$f_1(\tilde{\Delta})=\bigo{\tilde{\Delta}}$ and
$f_2(\tilde{n})=\bigo{\log^*\tilde{n}}$.  The following sketches our method
for obtaining a corresponding uniform MIS algorithm that runs in time
$\bigo{f(n,\Delta)}$.  For each integer $i$, set
$n_i=\max\sst{a\in\mathbf{N}}{f_1(a)\le2^i}$ and
$\Delta_i=\max\sst{a\in\mathbf{N}}{f_2(a)\le2^i}$.
In Iteration $i$, for $i=1,2,\ldots$, we first execute Algorithm $\cA$ using
the guesses $n_i$ and $\Delta_i$, but this time the execution lasts for
precisely $2\cdot2^{i}$ rounds. (The factor $2$ in the running time of an
iteration follows from the fact that the running time is the sum of two non-negative ascending
functions of two different parameters, namely $f_1(n)$ and $f_2(\Delta)$.)
Subsequently, we run the pruning algorithm~$\cP_{\MIS}$, and as before, when
the pruning algorithm terminates, we execute the next iteration on the
non-pruned nodes.  Now, let $s$ be the integer such that
$2^{s-1}<f(n,\Delta)\leq 2^s$. By the definition, $n\leq n_s$ and
$\Delta\leq \Delta_s$. Hence, the application of $\cA$ in Iteration $s$ uses
guesses that are indeed good. This execution of $A$ is completed before the
prescribed deadline of $2^{s+1}$ rounds expires because its running time is
at most $f_1(n_s)+f_2({\Delta_s})\leq 2^{s+1}$. Thus, the algorithm consists
of at most $s$ iterations. Since the running time of the whole execution is
dominated by the running time of the last iteration, the total running time
is $\bigo{2^{s+1}}=\bigo{f(n,\Delta)}$.

The above discussion shall be formalized in Theorem~\ref{simple}.
Before stating and proving it, though, we need one more concept, called
``sequence-number function'', which
gives a certain measure for the ``separation'' between the
variables in a function defined over $\Natural^\ell$.

\subsection{The General Framework}\label{sub:genframe}
Consider a  function $f\colon \Natural ^\ell\to \Reals^+$.
A \emph{set-sequence for $f$} is a sequence $(S_f(i))_{i\in\Natural}$
such that for every $i\in\Natural$,
\begin{enumerate}[label=\textup{(\roman{*})}, ref={(\roman{*})}]
    \item $S_f(i)$ is a finite subset (possibly empty) of $\Natural^\ell$; and
    \item if $\underline{y}\in\Natural^\ell$ and
        $f(\underline{y})\le i$, then $\underline{y}$ is
        dominated by a vector $\underline{x}$ that belongs to $S_f(i)$.
\end{enumerate}
The set-sequence $(S_f(i))_{i\in\Natural}$ is \emph{bounded} if
there exists a positive number $c$ such that
\[
\forall  i\in\Natural,\,
\forall \underline{x}\in S_f(i),\quad f(\underline{x})\le c\cdot i.
        \]
The constant $c$ is referred to as the \emph{bounding constant of
$(S_f(i))_{i\in\Natural}$}. Note that a set-sequence may contain empty sets.

A function $s_f\colon \Natural \to \Natural$
is a \emph{sequence-number function for $f$} if
\begin{enumerate}[label=\textup{(\arabic{*})}, ref={(\arabic{*})}]
    \item\label{snf:1} $s_f$ is moderately-slow; and
    \item\label{snf:2} there exists a bounded set-sequence
    $(S_f(i))_{i\in\Natural}$ for $f$ such that
\[
\forall i\in\Natural,\quad\abs{S_f(i)}\le s_f(i).
\]
\end{enumerate}

For example, consider the case where $f\colon \Natural ^\ell\to \Reals$ is additive,
i.e., $f(x_1,\ldots, x_\ell)=\sum_{k=1}^\ell f_k(x_k)$, where
$f_1,\ldots,f_\ell$ are non-negative ascending functions.
Here, the constant function $1$ is a sequence-number function for $f$.
Indeed, for $i\in\Natural$, let $S_f(i)=\{\underline{x}\}$,
where the $k$-th coordinate of $\underline{x}$ is defined to be the largest integer
$y$ such that $f_k(y)\leq i$ (if such an integer $y$ exists,
otherwise, $S_f(i)$ is empty).
Hence, if $f(\underline{y})\le i$ then we deduce that
$f_k(\underline{y}_k)\le i$ as each of the functions $f_1,\ldots,f_\ell$ is
non-negative. Therefore, $\underline{x}$ dominates $\underline{y}$.
Consequently, $(S_f(i))_{i\in\Natural}$ is a set-sequence
for $f$, which is bounded since
\[
f(\underline{x})\le\sum_{k=1}^\ell
f_k(\underline{x}_k)\le \ell\cdot i,
\]
and $\ell$ does not depend on $i$ (the bounding constant $c$ is equal to
$\ell$ in this case).

As another example, consider the case where $f\colon \Natural ^2\to
\Reals$ is given by $f(x_1,x_2)= f_1(x_1)\cdot f_2(x_2)$,
where $f_1$ and $f_2$ are ascending functions taking values at least $1$.
Then, the function $s_f(i)=\lceil\log i\rceil+1$ is a sequence-number
for $f$.
Indeed, for $i\in\Natural$ let
$S_f(i)=\sst{(x_{1}^{j},x_{2}^{j})}{j\in\icc{0}{\lceil\log i\rceil}}$
where $x_{1}^{j}$ is the largest integer $y_1$ such that $f_{1}(y_1)\leq 2^j$ and
$x_{2}^{j}$ is the largest integer $y_2$ such that $f_{2}(y_2)\leq 2^{\log i-j+1}$ for each
$j\in\icc{0}{\lceil\log i\rceil}$ (if such integers $y_1$ and $y_2$ exist, otherwise
we do not define the pair $(x_{1}^{j},x_{2}^{j})$).
Again, a straightforward check ensures that $(S_f(i))_{i\in\Natural}$ is a
bounded set-sequence for $f$ with bounding constant $2$.
On the other hand, it is interesting to note  that not all functions have a bounded sequence-number function, as one can
see by considering the $\min$ function over $\mathbf{N}^2$.  The following observation summarizes to two aforementioned examples. \\

\begin{observ}\label{obs:sequence}
\begin{itemize}
\item
The constant function $1$ is a sequence-number function for any additive function.\\
\item
Let $f\colon \Natural ^2\to
\Reals$ be a function given by $f(x_1,x_2)= f_1(x_1)\cdot
f_2(x_2)$, where $f_1\ge 1$ and $f_2\ge 1$ are ascending  functions.
Then, the function $s_f(i)=\lceil\log i\rceil+1$ is a sequence-number function for $f$.
\end{itemize}
\end{observ}

We now give an explicit construction of a local algorithm $\pi$, which will be used
to prove the forthcoming theorems.

Consider a problem $\Prob$ and a family of instances $\cF$.
Assume that $\cP$ is a pruning algorithm for $\Prob$.
Let $\cA^\Gamma$ be a deterministic algorithm
for $\Prob$ and~$\cF$ depending on a set $\Gamma$ of parameters
$\para_1,\ldots,\para_\ell$.
In addition, fix an integer $c$ and
let $(S_i)_{i\in\Natural}$ be a family of (possibly empty) subsets
of $\Natural^{\ell}$.

The algorithm $\pi$ runs in iterations, each of which
can be seen as a uniform alternating algorithm that operates on the
configurations in $\cF$.

Fix $i\in\Natural$ and let us write
$S_i=\{\underline{x}^1,\ldots,\underline{x}^{J_i}\}$.
For every $j\in\icc{1}{J_i}$, consider the uniform algorithm
$\cA_{j,i}$ that consists of running $\cA^{\Gamma}$
with the vector of guesses $\underline{x}^j$ of $S_i$.
More precisely, the $k$-th coordinate of
$\underline{x}^j$ is used as a guess for $\para_k$ for $k\in\{1,\ldots,\ell\}$.
Now, we define $\cA_{j,i}'$ to be the algorithm $\cA_{j,i}$
restricted to $c\cdot 2^i$ rounds.

An iteration of $\pi$ consists of running the
uniform alternating algorithm for the sequence of uniform algorithms
$\{\cA'_{j,i}\}_{j\in\icc{1}{J_i}}$ and the pruning algorithm $\cP$.
A pseudocode description of Algorithm $\pi$ is given by
Algorithm~\ref{alg:pi}.

\begin{algorithm} 
\Begin{
$(S_f(i))_{i\in\Natural}\longleftarrow$ bounded set-sequence for $f$ corresponding to $s_f$\;
$c\longleftarrow$ bounding constant of $(S_f(i))_{i\in\Natural}$\;
$(G_{1},\inp_{1})\longleftarrow(G,\inp)$\;
\For{$i$ from $1$ to $\infty$}{
$S_i\longleftarrow S_f(2^i)$\;
$J_i\longleftarrow\abs{S_i}$\;
$(G_{1,i},\inp_{1,i})\longleftarrow(G_i,\inp_i)$\;
  \For{$j$ from $1$ to $J_i$}{
  $\cA_{j,i}'\longleftarrow\cA^{\Gamma}$ restricted to $c\cdot 2^i$ rounds run with vector guesses
  $\underline{x}^j$ of $S_i$\;
  $\out_{j,i}\longleftarrow\cA_{j,i}'(G_{j,i},\inp_{j,i})$\;
  $(G_{j+1,i},\inp_{j+1,i})\longleftarrow\cP(G_{j,i},\inp_{j,i},\out_{j,i})$\;
  }
  $(G_{i+1},\inp_{i+1})\longleftarrow(G_{J_i+1,i},\inp_{J_{i}+1,i})$\;
  }
}

\caption{The algorithm $\pi$.}\label{alg:pi}
\end{algorithm}

We are now ready to state and prove Theorem~\ref{simple}, which deals with
deterministic local algorithms.

\subsection{The Deterministic Case}\label{sub:deterministic}

Theorem~\ref{simple}
considers a single set $\Gamma$ of non-decreasing parameters $\para_1,\ldots,\para_{\ell}$,
and assumes that (1) the given non-uniform algorithm $\cA^{\Gamma}$ depends on ${\Gamma}$
and (2) the running time of $\cA^\Gamma$ is evaluated according to the parameters in $\Gamma$.
Recall that in such a case, we say that a function $f\colon\Natural^{\ell}\to\Reals^+$
\emph{upper bounds} the running time
of $\cA^{\Gamma}$ with respect to $\Gamma$  if the running time
$T_{\cA^{\Gamma}}(G,\inp)$ of $\cA^{\Gamma}$  for every $(G,\inp)\in\cF$ using a
collection of good guesses $\tilde{\Gamma}=\{\tp_1,\ldots,\tp_\ell\}$
for $(G,\inp)$ is at most $f(\tp_1,\ldots,\tp_{\ell})$.

\begin{theorem}
\label{simple}
Consider a problem $\Prob$ and a family of instances $\cF$.
Let $\cA^\Gamma$ be a deterministic algorithm
for $\Prob$ and~$\cF$ depending on a set $\Gamma$ of non-decreasing parameters.
Suppose that the running time of $\cA^\Gamma$ is bounded from above
by some function $f\colon\Natural^\ell\to\Reals^+$ where $\ell=\card{\Gamma}$.
Assume that there exists a sequence-number function  $s_f$ for $f$, and
 a $\Gamma$-monotone pruning algorithm $\cP$ for $\Prob$ and $\cF$.
Then there exists a \emph{uniform} deterministic  algorithm  for $\Prob$ and $\cF$
whose running time
is $\bigo{f^*\cdot s_{f}(  f^*)}$, where $f^*=f(\Gamma ^*)$.
\end{theorem}

\begin{proof}
Let $\para_1,\ldots,\para_{\ell}$ be the parameters in $\Gamma$.
Fix a bounded set-sequence $(S_f(i))_{i\in\Natural}$ for $f$ corresponding to $s_f$
and let $c$ be the bounding constant of $(S_f(i))_{i\in\Natural}$.
Set $S_i= S_f(2^i)$ and $J_i=\abs{S_i}$, hence
$J_i\le s_f(2^i)$.

The desired uniform algorithm is the algorithm $\pi$ (Algorithm~\ref{alg:pi}).
We shall prove that $\pi$ is correct and runs in time $\bigo{s_f(2^m)\cdot 2^m}$
over every configuration in $\cF$, where $m=\lceil \log f^* \rceil$.

Fix $i\in\Natural$ and let us write
$S_i=\{\underline{x}^1,\ldots,\underline{x}^{J_i}\}$.
Each iteration of the inner loop of $\pi$ is called
\emph{Sub-iteration}, while \emph{Iteration} is reserved for iterations of the
outer loop. As written in the pseudocode description of $\pi$ given
by Algorithm~\ref{alg:pi},
$(G_{j,i},\inp_{j,i})$ is the configuration over which $\pi$ operates
during Sub-iteration $j$ of Iteration $s$, for $j\in\icc{1}{J_i}$.

Let us prove that Algorithm $\pi$ is correct.
Fix a configuration $(G,\inp)$ and set $\para_r^*=\para_r(G,\inp)$ for $r\in\icc{1}{\ell}$.
We consider the operation of $\pi$ on $(G,\inp)$.
Setting $f^*= f(\para_1^*, \ldots, \para_\ell^*)$,
we know that $f^*$ is an upper bound on the running time of $\cA^{\Gamma}$ over $(G,\inp)$,
assuming that $\cA^{\Gamma}$ uses the vector $\Gamma^*$ of correct guesses
$\para^*_1,\ldots,\para^*_\ell$.
Let $s$ be the least integer such that
$f^*\le 2^s$.
By the definition, there exists $j^*\in\icc{1}{J_s}$,
such that $\underline{x}^{j^*}$ dominates
$(\para_{1}^*,\ldots,\para_{\ell}^*)$.

The monotonicity property of $\cP$ implies that\break
$\para_r(G_{j-1,i},\inp_{j-1,i})\ge
\para_r(G_{j,i},\inp_{j,i})$ for every $r\in\icc{1}{\ell}$.
Thus, we infer by induction on $k$ that $\para^*_r=\para_r(G,\inp)\ge
\para_r(G_{j,i},\inp_{j,i})$ for every $i\in\mathbf{N}$,
$j\in\icc{1}{J_i}$ and $r\in\icc{1}{\ell}$.

Now, let us  consider Iteration $s$ of $\pi$.
Assume that some nodes are still
active during Iteration $s$ of $\pi$, that is, $V(G_{s})$ is not empty.
Iteration $s$ of $\pi$  is composed
of $J_s$ sub-iterations. During Sub-iteration $j$,
the algorithm  $\cA'_{j,s};\cP$ is executed over
$(G_{s}^{j},\inp_{s}^{j})$.
We know  that $\para_r^*\ge\para_r(G_{j,s},\inp_{j,s})$ for every $j\in[1,J_s]$, and every
$r\in\icc{1}{\ell}$. So, in Sub-iteration $j^*$ of Iteration $s$,
we have $x_{j^*,r}\ge \para_r^*\ge\para_r(G_{j^*,s},\inp_{j^*,s})$ for every
$r\in\icc{1}{\ell}$.

Sub-iteration $j^*$ consists of first running Algorithm
$\cA'_{j^*,s}$, which amounts to running $\cA^\Gamma$ for $c\cdot 2^s$ rounds
using the vector of guesses $\underline{x}^{j^*}$
By the definition of $S_f(2^s)$, it follows that
$f(\underline{x}^{j^*})\leq c \cdot 2^s$. Hence,
this execution of Algorithm $\cA^\Gamma$ is actually completed by time $c\cdot 2^s$.
Furthermore, since $\underline{x}^{j^*}$ dominates
$(\para_1(G_{j^*,s},\inp_{j^*,s}),\ldots,\break\para_\ell(G_{j^*,s},\inp_{j^*,s}))$,
the vector of guesses used by Algorithm $\cA^\Gamma$ is good, and hence the algorithm outputs a vector
$\out_{s}^{j^*}$ satisfying  $(G_{j^*,s},\inp_{j^*,s},\out_{j^*,s})\in \Prob$.
By the solution detection property, the subsequent pruning algorithm (still in Sub-iteration $j^*$
of Iteration $s$)
selects $W_{j^*,s}=V(G_{j^*,s})$. By Observation~\ref{sequence}, it follows that $\pi$ is correct.

It remains to prove that the running time is $\bigo{s_f(f^*)\cdot f^*}$.
Let $T_0$ be the running time of $\cP$.
Observe that Iteration $i$ of $\pi$ takes at most $J_i(c\cdot2^i+T_0)$ rounds,
which is $\bigo{s_f(2^i)\cdot 2^i}$ rounds.
Since $\pi$ consists of at most $s$ iterations,
the running time of $\pi$ is bounded by
$\sum_{i=1}^{s}s_f(2^i)\cdot2^i$, which is $\bigo{s_f(2^s)\cdot2^s}$
because $s_f$ is non-decreasing. Moreover,
\[
\bigo{s_f(2^s)\cdot2^s}=\bigo{s_f(2\cdot f^*)\cdot2^s}=\bigo{s_f(f^*)\cdot f^*}
\]
since $2^{s-1}<f^*\le2^s$ and $s_f$ is moderately-slow (hence, in particular,
non-decreasing).
Therefore, the running time of $\pi$ is bounded by $\bigo{s_f(f^*)\cdot f^*}$.
\qed
\end{proof}

By Observation~\ref{obs:sequence}, the constant function $s_f=1$ is a
sequence number function for any additive function $f$.
Hence, Corollary~\ref{cor:main}\ref{thm:mm} follows directly by applying
Theorem~\ref{simple} to the maximal matching algorithm of
Hanckowiak \emph{et al.}~\cite{HKP01}, and using Observation~\ref{claim:universal-mm}.

In addition, using Observation~\ref{claim:universal-ruling},
Theorem~\ref{simple} allows us
to transform each of the  MIS algorithms in~\cite{BaEl09,Kuh09,PaSr96}
into a uniform one with asymptotically the same time complexity.
We thus obtain the following corollary.\\

\begin{corollary}\label{cor:MIS} Consider the family $\cF$ of all graphs.
\begin{itemize}
\item There exists a uniform deterministic  $\MIS$ algorithm  for  $\cF$ running  in time $\bigo{\Delta+\log^*\! n}$.\\
\item
There exists a uniform deterministic $\MIS$ algorithm for $\cF$ running in time $2^{\bigo{\sqrt{\log n}}}$.
\end{itemize}
\end{corollary}

Recall that Barenboim and Elkin~\cite{BaEl10a} devised, for every $\delta>0$, a (non-uniform) deterministic $\MIS$ algorithm for the family of all graphs running in time $f(a,n)=\bigo{a+a^{\delta}\log n}$.
Fix $\epsilon\in(0,1)$ and consider the family $F_{\text{large}}$ of graphs with arboricity $a>\log^{1+\epsilon/2} n$. It follows from~\cite{BaEl10a} (applied with, e.g., $\delta=\epsilon/3$), that there exists a (non-uniform) deterministic $\MIS$ algorithm for
$F_{\text{large}}$ running in time $\bigo{a}$. Hence, using Observation~\ref{claim:universal-ruling} and
Theorem~\ref{simple}, we obtain a uniform deterministic $\MIS$ algorithm  for
$F_{\text{large}}$ running in $\bigo{a}$ time.

Next, let $F_{\text{med}}$ be  the family of graphs with arboricity $a$ such that $\log^{1/3} n< a\leq \log^{1+\epsilon/2} n$.
Since $a \leq \log^{1+\epsilon/2}n$, it follows that $a^{1-\epsilon/2} < \log n$, and hence, $a<a^{\epsilon/2}\log n$.
By~\cite{BaEl10a}, applied with $\delta=\epsilon/2$, there exists a deterministic $\MIS$ algorithm for
$F_{\text{med}}$ running in time $f_{\text{med}}=\bigo{a^{\epsilon/2}\log n}$.
Note that by Observation~\ref{obs:sequence}, the sequence number for
$f_{\text{med}}$ is $s_{f_{\text{med}}}(f_{\text{med}})=\bigo{\log
f_{\text{med}}}=\bigo{\log\log n}$.
Hence, by combining Observation~\ref{claim:universal-ruling} and
Theorem~\ref{simple}, we obtain a uniform $\MIS$ algorithm for $F_{\text{med}}$
running in time $\bigo{a^{\epsilon/2}\log n\log\log n}=\bigo{a^{\epsilon}\log n}$.
(This last equality follows from the fact that $\log^{1/3} n< a$.)\footnote{In fact,
we could have used  in the definition of $F_{\text{med}}$ any small constant instead of $1/3$, but
$1/3$ is sufficiently good for our purposes as, anyway, this result will
be combined with better results for $a=\lo{\sqrt{\log n}}$,
which shall be established later on, in Corollary \ref{cor:arboricity}.}
Summarizing the above discussion, we obtain the following.

\begin{corollary}\label{cor:MIS-arb} For every $\epsilon>0$, there exists the following uniform deterministic $\MIS$ algorithm:
\begin{itemize}
\item
For the family $F_{\text{large}}$, running in $O(a)$ time,\\
\item
For the family $F_{\text{med}}$, running in $\bigo{a^{\epsilon}\log n}$ time.
\end{itemize}
\end{corollary}

\subsection{The Randomized Case}\label{sub:randomized}

We now show how to extend Theorem~\ref{simple} to the rando\-mi\-zed setting. More
specifically, we replace the given non-uniform deterministic algorithm
of Theorem~\ref{simple} by a non-uniform weak Monte-Carlo algorithm $\cA^\Gamma$
and produce a uniform Las Vegas one.
This transformer is more sophisticated than the one given in  Theorem~\ref{simple},
and requires the use of sub-iterations for bounding the expected running time and probability of success of the resulting Las-Vegas algorithm.

\begin{theorem}
\label{thm:randomized}
Consider a problem $\Prob$ and a family of instances $\cF$.
Let $\cA^\Gamma$ be a weak Monte-Carlo algorithm
for $\Prob$ and~$\cF$ depending on a set
$\Gamma$ of non-decreasing parameters. Suppose that the running time of
$\cA^\Gamma$ is bounded from above
by some function $f\colon\Natural^\ell\to\Reals^+$, where $\ell=\card{\Gamma}$.
Assume that there exists a sequence-number function $s_f$ for $f$, and
a $\Gamma$-monotone pruning algorithm $\cP$ for $\Prob$ and $\cF$.
Then there exists a \emph{uniform} Las Vegas algorithm  for $\Prob$ and $\cF$
whose expected running time
is $\bigo{f^*\cdot s_{f}(  f^*)}$, where $f^*=f(\Gamma ^*)$.
\end{theorem}

\begin{proof}
Let $\para_1,\ldots,\para_\ell$ be the parameters in $\Gamma$.
Let $T_0$ be the running time of the pruning algorithm  $\cP$,
and let $\cA^{\Gamma}$ be the given weak Monte-Carlo algorithm.
To simplify the notations, we assume that the success guarantee $\rho$ of $\cA^{\Gamma}$
 is $1/2$.

\begin{algorithm}
\Begin{
$(S_f(i))_{i\in\Natural}\longleftarrow$ bounded set-sequence for $f$ corresponding to $s_f$\;
$c\longleftarrow$ bounding constant of $(S_f(i))_{i\in\Natural}$\;
$(G_1,\inp_1)\longleftarrow(G,\inp)$\;
\For{$i$ from $1$ to $\infty$}{
    \For{$j$ from $1$ to $i$}{
        $S_j\longleftarrow S_f(2^j)$\;
        $J_j\longleftarrow\abs{S_j}$\;
        $(G_{1,j},\inp_{1,j})\longleftarrow(G_i,\inp_i)$\;
        \For{$k$ from $1$ to $J_j$}{
            $\cA_{k,j}'\longleftarrow\cA^{\Gamma}$ restricted to $c\cdot 2^j$ rounds run with vector
                guesses $\underline{x}^k$ of $S_j$\;
            $\out_{k,j}\longleftarrow A_{k,j}'(G_{k,j},\inp_{k,j})$\;
            $(G_{k+1,j},\inp_{k+1,j})\longleftarrow\cP(G_{k,j},\inp_{k,j},\out_{k,j})$\;
        }
        $(G_{j+1},\inp_{j+1})\longleftarrow(G_{J_j+1,j},\inp_{J_{j}+1,j})$\;
    }
}
}

\caption{The algorithm $\tau$ in the proof of
Theorem~\ref{thm:randomized}.}\label{alg:tau}
\end{algorithm}

The desired uniform algorithm $\tau$ runs in iterations, where Iteration $i$
consists of running the first $i$ iterations of the algorithm $\pi$ defined
in Subsection~\ref{sub:genframe}.
A pseudocode description of Algorithm
$\tau$ is given by Algorithm~\ref{alg:tau}. Similarly as in the proof
of Theorem~\ref{simple}, the word ``Iteration'' is reserved for the iterations
of the outer loop of $\tau$, while ``Sub-iteration'' is used for the iterations
of the middle loop of $\tau$.

For each positive integer $i$, let $\beta_i$ be the number of rounds used in
Iteration $i$ of $\tau$.
Analogously to the proof of Theorem~\ref{simple}, we infer that
$\beta_i=\bigo{s_f(2^i)\cdot 2^i}$.
Let $\alpha_i$ be the number of rounds used during the first $i$ iterations of
$\tau$. We thus have $\alpha_i=\sum_{k=1}^{i}\beta_k$, which is
$\bigo{s_f(2^i)\cdot 2^i}$.

It follows using similar arguments to the ones given in the proof of Theorem~\ref{simple}, that
if~$\tau$ outputs, then the output vector $\out$ is a solution, i.e. $(G,\inp,\out)\in\Prob$.

It remains to bound the running time of $\tau$.
We consider the random variable $\run_{\tau}(G,x)$ that stands for ``the
running time of $\tau$ on $(G,\inp)$''.
For every integer $i$, let $\rho_i$ be the probability that $V(G_{i})\neq \emptyset$ and
$V(G_{i+1})=\emptyset$, that is, $\rho_i$ is the probability that the last active node
becomes inactive precisely during Iteration $i$ of $\tau$. In other words,
\[
\rho_i=\pr\left(\run_{\tau}(G,x)\in\icc{\alpha_{i-1}+1}{\alpha_{i}}\right).
\]

Setting $f^*= f(\para_1^*, \ldots, \para_\ell^*)$, we know that $f^*$ is an upper
bound on the running time of $\cA^{\Gamma}$ over $(G,\inp)$, assuming that
$\cA^{\Gamma}$ uses the collection $\Gamma^*$ of correct guesses
$\para^*_1,\ldots,\para^*_\ell$.
Consider the smallest integer $s$ such that
$f^*\le 2^s$.

Since $s_f$ is moderately-slow, there is a constant $K$ such that
$\alpha_{i+1}\le K\cdot\alpha_i$ for every positive integer $i$.
In particular, $\alpha_{s+i}\le K^{i}\cdot\alpha_s$, and hence
\begin{align*}
\Ee(\run_{\tau}(G,x))&\le
\alpha_s\cdot\pr\left(\run_{\tau}(G,\inp)\le\alpha_s\right)+\sum_{i=1}^{\infty}\alpha_{s+i}\cdot
\rho_{s+i}\\
&\le\alpha_s+\alpha_s\sum_{i=1}^{\infty}K^{i}\cdot
\rho_{s+i}.
\end{align*}
Our next goal is to bound $\rho_{s+i}$ from above. For a positive integer $r$,
let $\chi_r$ be the event that $V(G_{r+1})\neq\emptyset$,
that is, none of $\cC_1,\ldots,\cC_r$ output the empty configuration and thus,
there is still an active node at the beginning Iteration $r+1$ of $\tau$.
Thus,
$\rho_{s+i}\le  \pr(\chi_{s+i-1})$.

Recall that we assume that the success guarantee  of
$\cA^{\Gamma}$ is $1/2$. Therefore, using similar analysis as in the proof of Theorem \ref{simple}, it follows that
for every positive integer $k$, the probability that an application of $\cB_{s+k-1}$ (in particular, during iteration $s+i-1$) does not output the empty configuration is at most $1/2$.
As a result,
\[
\rho_{s+i}\le\pr(\chi_{s+i-1})\le\prod_{j=1}^{i}2^{-j}=2^{- (i^2+i)/2}.
\]
Therefore,
\begin{align*}
\Ee(\run_{\tau}(G,x))&\le\alpha_s\left(1+\sum_{i=1}^{\infty}K^i\cdot2^{-(i^2+i)/2}\right)\\
&=\bigo{\alpha_s}
=\bigo{f^*\cdot s_f(f^*)}.
\end{align*}
\qed
\end{proof}

Corollary~\ref{cor:main}\ref{thm:rs} follows by applying Theorem~\ref{thm:randomized}
to the ruling set algorithm of Schneider and Wattenhofer~\cite{ScWa10a},
and using the pruning algorithm given by Observation~\ref{claim:universal-ruling}.

\subsection{The General Theorem}\label{sub:general}
Some complications arise when the correctness of the given non-uniform
algorithm relies on the use of a set of parameters $\Gamma$ while its running time is evaluated
with respect to another set of parameters $\Lambda$.
For example, it may be the case that an upper bound on a parameter $\para$ is required for the correct operation of an algorithm, yet the running time of the algorithm does not depend on $\para$. In this case, it may not be clear how to choose the guesses for $\para$.
(This occurs, for example, in the MIS algorithms of Barenboim and Elkin~\cite{BaEl10b}, where the knowledge of $n$ and the arboricity $a$ are required, yet the running time $f$ is a function of $n$ only.)
Such complications can be solved when there is some relation between the parameters in $\Gamma$ and
those in $\Lambda$; specifically, when $\Gamma$ is weakly-dominated by $\Lambda$.
(The definition of weakly-dominated is given in Section~\ref{sec:preliminaries}.)
This issue is handled in the following theorem, which extends both Theorem~\ref{simple} and Theorem~\ref{thm:randomized}.

\begin{theorem}
\label{thm:parameters}
Consider a problem $\Prob$, a family of instances $\cF$ and two
sets of non-decreasing parameters $\Gamma$ and
$\Lambda$,  where $\Gamma$ is weakly-dominated by
$\Lambda$. Let
$\cA^\Gamma$ be a deterministic (respectively, weak Monte-Carlo) algorithm
 depending on $\Gamma$ whose running time  is upper  bounded by some function
 $f\colon\Natural^\ell\to\Reals^+$, where $\ell=\card{\Lambda}$.
Assume that there exists a  sequence-number function  $s_f$ for $f$, and
 a $\Lambda\cup \Gamma$-monotone pruning algorithm $\cP$ for $\Prob$ and $\cF$.
Then there exists a \emph{uniform} deterministic (resp., Las Vegas) algorithm  for $\Prob$ and $\cF$
whose running time on every configuration $(G,\inp)\in\cF$ is
$\bigo{f^*\cdot s_{f}(  f^*)}$, where $f^*= f(\Lambda^*(G,\inp))$.
\end{theorem}

\begin{proof}
First, we consider the case where $\Gamma\subseteq\Lambda$ and next
the general case.

Assume that
$\Lambda=\{\para_1,\ldots,\para_{\ell}\}$ and
$\Gamma=\{\para_1,\ldots,\para_r\}$, where $r\leq\ell$.
Then, let us simply impose that $A^{\Gamma}$ also requires estimates for the parameters
$\para_{r+1},\ldots,\para_{\ell}$, that is,
the operation of $A^{\Gamma}$ requires such estimates but actually ignores them after obtaining them.
This way, we obtain an algorithm
$A^{\Lambda}$ depending on $\Lambda$. Since $f$ is non-decreasing,
$f(\para^*_{1},\ldots,\para^*_{\ell})\leq
 f(\para^*_{1},\ldots,\para^*_{r},\tp_{r+1},\ldots,\tp_{\ell})$, where $\tp_i$ is a  good
guess for every $i\in\icc{r+1}{\ell}$. Hence, the running time of Algorithm
$A^{\Lambda}$ is also bounded by~$f$, so the conclusion follows by applying
Theorems~\ref{simple} and~\ref{thm:randomized}.

Now, let $\para_1,\ldots,\para_r$ and $\q_1,\ldots,\q_{\ell}$ be the parameters in
$\Gamma$ and $\Lambda$, respectively.
Recall that $r'\in\icc{0}{\min\{r,\ell\}}$ is such that
$\{\para_{r'+1},\para_{r'+2},\ldots, \para_r\}\cap
\{\q_{r'+1},\q_{r'+2},\ldots,\q_\ell\}=\emptyset$ and $\para_i=\q_i$ for every
$i\in\icc{1}{r'}$.
Set $t= r-r'$.
As $\Gamma$ is weakly-dominated by $\Lambda$,
there exists a function $h:\icc{1}{t}\to\icc{1}{\ell}$ and, for each $j\in\icc{1}{t}$,
an ascending function $g_j$ such that $g_j(\para_{r'+j}(G,\inp))\leq \q_{h(j)}(G,\inp)$ for
every configuration $(G,\inp)\in\cF$.
For every real number $x$, we set $g_j^{-1}(x)=\min g_j^{-1}(\{x\})$.
Since $g_j$ is ascending, $g_j^{-1}(x)\ge g_j^{-1}(y)$ whenever $x\ge y$.

Let $\Lambda'=\Lambda \cup \Gamma=\{\q_1,\ldots,\q_\ell,
\para_{r'+1},\ldots,\para_r\}$, and recall that
$f\colon\Natural^\ell\to\Reals^+$ is the (non-decreasing)
function bounding the running time of $\cA^{\Gamma}$.
We define a new function $f':\Natural^{\ell+t}\to\Reals$ by setting
\[
f'(x_1,\ldots,x_\ell,y_{1},\ldots,y_t)= f(z_1,\ldots,z_\ell),
\]
where for each $i\in\icc{1}{\ell}$,
\[
z_i=\max\left(\{x_i\}\cup\sst{g_k(y_k)}{k\in h^{-1}(\{i\})}\right).
\]
Let $s_f$ be a sequence-number function for $f$ and let
$(S_f(i))_{i\in\Natural}$ be a corresponding bounded set-sequence with bounding constant $c$.

We assert that $s_f$ is also a sequence-number function of $f'$ and admits a
corresponding bounded set-sequence with bounding constant $c$.
To see this, we first define
for $i\in\Natural$ a set $S_{f'}(i)$ with
$\abs{S_{f'}(i)}=\abs{S_f(i)}$ as follows.
For each $(x_1,\ldots,x_\ell)\in S_f(i)$, let
$S_{f'}(i)$ contain $(x_1,\ldots,x_\ell,
y_1,\ldots,y_t)$, where $y_j= g^{-1}_j(x_{h(j)})$ for
$j\in\icc{1}{t}$.
Observe that $g_j(y_j)=x_{h(j)}$ for every $j\in\icc{1}{t}$. Hence,
$f'(x_1,\ldots,x_\ell,y_1,\ldots,y_t)=f(x_1,\ldots,x_\ell)$ if
$(x_1,\ldots,x_\ell,y_1,\ldots,y_t)\in S_{f'}(i)$.

This observation directly implies that $f'(\underline{x}')\le c\cdot i$ if
$\underline{x}'\in S_{f'}(i)$, since
$f(\underline{x})\le c\cdot i$ if $\underline{x}\in S_f(i)$.
Now, assume that $f'(\underline{x})\le i$ for some
$\underline{x}=(x_1,\ldots,x_\ell,y_1,\ldots,y_\ell)\in\Natural^{\ell+t}$. Then,
$f(z_1,\ldots,z_\ell)\le i$, where $z_i$ is given by the definition of $f'$.
Consequently,
there exists a vector $\underline{\tilde{z}}\in S_f(i)$ that dominates
$(z_1,\ldots,z_\ell)$. Moreover,
\[
\underline{\tilde{z}}'=(\underline{\tilde{z}}_1,\ldots,\underline{\tilde{z}}_\ell,g^{-1}_1(\underline{\tilde{z}}_{h(1)}),\ldots,g^{-1}_t(\underline{\tilde{z}}_{h(t)}))\in
S_{f'}(i).
\]
Therefore, if $j\in\icc{1}{\ell}$ then
$(\underline{z'})_j=\underline{\tilde{z}}_j\ge z_j\ge x_j$, and if $j\in\icc{1}{t}$ then
$g_j((\underline{z}')_{\ell+j})=\underline{\tilde{z}}_{h(j)}\ge
z_{h(j)}\geq g_j({y}_{j})$, so $(\underline{z'})_{\ell+j}\ge y_j$,
as $g_j$ is ascending.
This finishes the proof of the assertion.

Since $\Gamma\subseteq \Lambda'$, we know that there exists a uniform local deterministic
(respectively, randomized Las Vegas) algorithm
$\cA$ for $\Prob$ and $\cF$
such that the (respectively, expected) running time of $\cA$ over any configuration
$(G,\inp)\in \cF$ is $\bigo{f'^*\cdot s_{f'}(f'^*)}=\bigo{f'^*\cdot s_{f}(  f'^*)}$, where
$f'^*=
f(\q^*_1,\ldots,\q^*_\ell,\para^*_{r'+1},\ldots,\para^*_r\})$.
The fact that $f'$ is non-decreasing implies that
\[
f'^*\leq f'(\q^*_1,\ldots,\q^*_\ell,g_1^{-1}(\q^*_{h(1)}),\ldots,g_t^{-1}(\q^*_{h(t)}))
=f^*.
\]
As $s_f$ is non-decreasing, the (respectively, expected) running time of $\cA$ is bounded by
$\bigo{f^*\cdot s_{f}(f^*)}$, as desired.
\qed
\end{proof}

Applying Theorem~\ref{thm:parameters} to the
work of Barenboim and Elkin~\cite{BaEl10b} (see Theorem 6.3 therein) with $\Gamma=\{a,n\}$ and
$\Lambda=\{n\}$ yields the following result, since $a\le n$.

\begin{corollary}\label{cor:arboricity}
The following  uniform deterministic algorithms solving $\MIS$ exist :
\begin{itemize}
\item
For the family of  graphs with arboricity $a=\lo{\sqrt{\log n}}$, running in time $\lo{\log n}$,  \\
\item
For any constant $\delta\in\ioo{0}{1/2}$, for  the family of
graphs with arboricity
$a=\bigo{\log^{1/2-\delta} n}$, running in time $\bigo{\log n/\log\log n}$.\\
\end{itemize}
\end{corollary}

\subsection{Running as Fast as the Fastest Algorithm}\label{sec:min}

To illustrate the topic of the next theorem, consider the  non-uniform  algorithms for MIS for general graphs, namely, the
algorithms of Barenboim and Elkin~\cite{BaEl09} and that of Kuhn~\cite{Kuh09}, which run in time
$\bigo{\Delta+\log^* n}$ and use the knowledge of $n$ and $\Delta$, and
the algorithm of Panconesi and Srinivasan~\cite{PaSr96}, which runs in time
$2^{\bigo{\sqrt{\log n}}}$ and requires the knowledge of $n$. Furthermore, consider the MIS algorithms of
Barenboim and Elkin in ~\cite{BaEl10a,BaEl10b}, which
are very efficient for graphs with a small arboricity $a$.
If $n$, $\Delta$ and $a$ are contained in the inputs of all nodes,
then one can compare the running times of these
algorithms and use the fastest one. That is, there exists a non-uniform
algorithm $\cA^ {\{n,\Delta,a\}}$ that runs in time
$T(n,\Delta,a)=\min\{g(n),h(\Delta,n), f(a,n)\}$, where  $g(n)=2^{\bigo{\sqrt{\log n}}}$, $h(\Delta,n)=\bigo{\Delta+\log^* n}$, and $f(a,n)$ is defined as follows: $f(a,n)=\lo{\log n}$ for graphs of arboricity $a=\lo{\sqrt{\log n}}$,
$f(a,n)=\bigo{\log n/\log\log n}$ for arboricity
$a=\bigo{\log^{1/2-\delta} n}$, for some constant $\delta\in\ioo{0}{1/2}$; and
otherwise: $f(a,n)=\bigo{a+a^{\epsilon}\log n}$, for arbitrary small constant $\epsilon>0$.

Unfortunately, the theorems established so far do not allow us to transform
$\cA^ {\{n,\Delta, a\}}$ into a uniform algorithm---the reason being that the
function $T(n,\Delta,a)$ bounding the running time does not have a sequence number.
On the other hand, as mentioned in Corollary~\ref{cor:MIS},
Theorem~\ref{simple} does allow us to transform each of the algorithms in \cite{BaEl09,Kuh09,PaSr96}
into a uniform MIS algorithm, with time complexity
$\bigo{\Delta+\log^* n}$ and $2^{\bigo{\sqrt{\log n}}}$, respectively.
Moreover, Corollaries  \ref{cor:MIS-arb} and~\ref{cor:arboricity} show that Theorems~\ref{simple} and~\ref{thm:parameters} allow us to transform the algorithms  in ~\cite{BaEl10a,BaEl10b} to uniform algorithms
that (over the appropriate  graph families), run as fast as the corresponding non-uniform algorithms .
Nevertheless,
unless $n$, $\Delta$ and $a$ are provided as inputs to the nodes, it is not clear how to
obtain from these transformed algorithms a uniform algorithm running~in~time
$T(n,\Delta,a)$. The following theorem solves this problem.

\begin{theorem}
\label{thm:min}
Consider a problem $\Prob$ and a family of instances $\cF$. Let $k$ be a
positive integer and
let $\Lambda_1,\ldots, \Lambda_k$ be $k$ sets of non-decreasing parameters.
Let $\cP$ be
a $(\Lambda_1\cup \cdots\cup \Lambda_k)$-monotone  pruning algorithm for $\Prob$ and $\cF$.
For $i\in\{1,2,\cdots,k\}$, consider a uniform algorithm $\cU_i$ whose running time
is bounded with respect to $\Lambda_i$ by a function $f_i$.
Then there is a uniform algorithm with running time $\bigo{f_{\min}}$,
where $f_{\min}={\min\{f_1(\Lambda_1^*), \ldots, f_k(\Lambda_k^*)\}}$.
\end{theorem}

\begin{proof}
Clearly, it is sufficient to prove the theorem for the case $k=2$.
The basic idea behind the proof of theorem above is to run in iterations,
such that each iteration $i$ consists of running the quadruple $(\cU_1;\cP;\cU_2;\cP)$,
where $\cU_1$ and $\cU_2$ are executed for precisely $2^i$ rounds each.
Hence, a correct solution will be produced in Iteration
$s=\lceil\log f_{\min}\rceil$ or
before. Since each iteration $i$ takes at most $\bigo{2^i}$ rounds (recall
that the running time of $\cP$ is constant), the running time is $\bigo{f_{\min}}$.

Formally, we define a sequence of uniform algorithms $(\cA_i)_{i\in\mathbf{N}}$
as follows. For $i\in\mathbf{N}$, set $\cA_{2i+1}=\hat{\cU_1}$ and
$\cA_{2i+2}=\hat{\cU_2}$, where $\hat{\cU_j}$ is
$\cU_j$ restricted to $2^{i}$ rounds for $j\in\{1,2\}$.
Let $\pi$ be the uniform alternating algorithm with respect to
$(\cA_i)_{i\in\mathbf{N}}$ and $\cP$,
that is $\pi=\cB_{1};\cB_{2};\cB_{3};\cdots$ where
$\cB_{2i+j}=\hat{\cU_j};\cP$ for every $i\in\mathbf{N}$ and every
$j\in\{1,2\}$.
Letting $T_0$ be the running time of $\cP$,
the running time of $\cB_i$ is at most $2^{i}+T_0$, for every $i\in\mathbf{N}$.

Consider an instance $(G,\inp)\in\cF$. For each $(\para,\q)\in
\Lambda_1\times\Lambda_2$,
let $\para^*=\para(G,\inp)$ and $\q^*=\q(G,\inp)$.
Algorithm $\cB_i$ operates on the configuration $(G_i,\inp_i)$.
Let $\para\in \Lambda_1\cup \Lambda_2$. Because $\cP$  is monotone with
respect to $\Lambda_1\cup \Lambda_2$, it follows by induction on $i$
that $\para^*\geq \para(G_i,\inp_i)$.
Hence, the running time of $\cU_j$
over $(G_i,\inp_i)$ is bounded from above by
$f_j(\Lambda_j^*)$ for every $i\in\mathbf{N}$ and each $j\in\{1,2\}$.
Thus, $V(G_{2s+2})=\emptyset$ for the smallest $s$ such that
$2^s\geq f_{\min}$. In other words, $\pi=\cB_{1};\cB_{2};\cdots;\cB_{2s+1}$.
Consequently, by Observation~\ref{sequence}, Algorithm $\pi$ correctly solves $\Prob$ on $\cF$
and, since $\cB_i$ runs in at most $2^{\lceil i/2\rceil }+T_0$ rounds,
the running time of $\pi$ is $\bigo{2^s}=\bigo{f_{\min}}$,
as asserted.
\qed
\end{proof}

Now, we can combine  Theorem~\ref{thm:min}
with  Corollaries~\ref{cor:MIS-arb} and~\ref{cor:arboricity}, and establish a uniform algorithm for MIS that runs in time $f(a,n)$.
Combining this algorithm with Corollary~\ref{cor:MIS}, and applying once more Theorem~\ref{thm:min}
yields Corollary~\ref{cor:main}\ref{thm:mis}.

\section{Uniform Coloring Algorithms}\label{sec:coloring}

In general, we could not find a way to directly apply our transformers (e.g., the one given by Theorem~\ref{thm:parameters}) for the coloring problem.  The main reason is that we could
not find an efficient pruning algorithm for the coloring problem. Indeed, consider for example the $O(\Delta)$-coloring problem. The checking property of a pruning algorithm
requires that, in particular, the nodes can locally decide whether they belong to a legal configuration.
While locally checking that neighboring nodes have distinct colors is easy,
knowing whether a color is in the required range, namely,
$\icc{1}{\bigo{\Delta}}$, seems difficult as the nodes do not know
$\Delta$. Moreover, the gluing property seems difficult to tackle also:
after pruning a node with color $c$,  none of its unpruned neighbors can be colored in color $c$.
In other words, a correct solution on the non-pruned subgraph may not glue well with the pruned subgraph.

Nevertheless, we show in this section that several relatively general transformers can be used to obtain uniform coloring algorithms from non-uniform one.
We focus on standard coloring problems in which the required number of colors is given as a function of $\Delta$.

\subsection{Uniform $(\Delta+1)$-coloring Algorithms}
A standard trick (cf., \cite{Lin92,Lub86}) allows us to
transform an efficient (with respect to $n$ and $\Delta$) MIS algorithm for general graphs into one for
$(\Delta+1)$-coloring (and, actually, to the more general \emph{maximal
coloring} problem defined by Luby~\cite{Lub86}).
The general idea is based on the observation that $(\Delta+
1)$-colorings of $G$ and maximal independent sets of $G'= G \times K_{\Delta+1}$ are in
one-to-one correspondence. More precisely, and avoiding the use of $\Delta$,
the graph $G'$ is constructed from $G$ as follows.
For each node
$u\in V(G)$, take a clique $C_u$ of size $\deg_G(u)+1$ with nodes
$u_1,\ldots,u_{\deg_G(u)+1}$.  Now, for each $(u,v)\in E(G)$ and each
$i\in\icc{1}{1+\min\{\deg_G(u),\deg_G(v)\}}$, let $(u_i,v_i)\in E(G')$.
The graph $G'$ can be constructed by a local algorithm without using any
global parameter. It remains to observe the existence of a natural one-to-one
correspondence between the maximal independent sets of $G'$ and the
$(\deg_G+1)$-colorings of $G$, that is, the colorings of $G$ such that each node
$u$ is assigned a color in $\icc{1}{\deg_G(u)+1}$.

To see this, first
consider a $(\deg_G+1)$-coloring $c$ of $G$. Set
\[
X=\sst{u_i\in V(G')}{c(u)=i}.
\]
Then, no two nodes in $X$ are adjacent in $G'$. Moreover, a node that does
not belong to $X$ has a neighbor in $X$ since $X$ contains
a vertex from each clique $C_u$ for $u\in V(G)$. Therefore, $X$ is a MIS of
$G'$.

Conversely, let $X$ be a MIS of $G'$. We assert that $X$ contains a node
from every clique $C_u$ for $u\in V(G)$. Indeed, suppose on the contrary
that $X\cap V(C_u)=\emptyset$ for a node $u\in V(G)$. By the definition
of a MIS, every vertex $u_i\in V(C_u)$ has a neighbor $v(u_i)$ that belongs to
$X$. Since a clique can contain at most one node in $X$ and $v(u_i)\neq
v(u_j)$ whenever $i\neq j$, we deduce that at least $\abs{C_u}$ cliques $C_v$
with $v\neq u$ contain a node that has a neighbor in $C_u$. This
contradicts the definition of $G'$, since $\abs{C_u}=\deg_G(u)+1$.
Thus, setting $c(u)$ to be the index $i\in\{1,\ldots,\deg_G(u)+1\}$
such that $u_i\in X$ yields a $(\deg_G+1)$-coloring of $G$.

Therefore, we
obtain Corollary~\ref{cor:main}\ref{thm:color-delta} as a direct consequence of
Corollary~\ref{cor:main}\ref{thm:mis}.

\subsection{Uniform Coloring with More than $\Delta+1$ Colors}

We now aim to provide a transformer taking as input an efficient non-uniform coloring algorithm that uses $g(\Delta)$ colors (where  $g(\Delta)>\Delta$) and produces
an efficient uniform coloring algorithm that uses $\bigo{g(\Delta)}$ colors.
We begin with the following definitions.

An \emph{instance} for the coloring problem is a pair $(G,\inp)$ where $G$
is a graph and $\inp(v)$ contains a color $c(v)$ such that the collection
$\sst{c(v)}{v\in V(G)}$ forms a coloring of $G$. (The color $c(v)$ can be the
identity $\id(v)$.) For a given family $\cal{G}$ of graphs, we define
$\cF(\cal{G})$ to be the collection of instances $(G,\inp)$ for the
coloring problem, where $G\in\cal{G}$.

Many coloring algorithms consider the identities as
colors, and relax the assumption that the identities are unique by replacing it
with the weaker requirement that the set of initial colors forms a coloring.
Given an instance~$(G,\inp)$, let $m=m(G,\inp)$ be the maximal identity.
Note that $m$ is a graph-parameter.

Recall the $\lambda(\tilde{\Delta}+1)$-coloring algorithms
designed by Barenboim and Elkin~\cite{BaEl09} and Kuhn~\cite{Kuh09}
(which generalize the $\bigo{\tilde{\Delta}^2}$-coloring algorithm of
Linial~\cite{Lin92}).
We would like to point out that, in fact, everything works similarly in these
algorithms if one replaces  $n$ with $m$.
That is, these $\lambda(\tilde{\Delta}+1)$-coloring algorithms
can be viewed as requiring $m$ and  $\Delta$ and running
in time $\bigo{{\tilde{\Delta}/\lambda+\log^* \tilde{m}}}$.
The same is true for the edge-coloring algorithms of Barenboim and Elkin~\cite{BaEl11}.

The following theorem implies that these algorithms can be transformed into uniform ones.
In the theorem, we consider two sets $\Gamma$ and $\Lambda$ of non-decreasing graph-parameters such that
\begin{itemize}
\item[(1)] $\Gamma$ is weakly-dominated by $\Lambda$; and
\item[(2)] $\Gamma\subseteq\{\Delta, m\}$.
\end{itemize}
Two such sets of parameters are said to be \emph{related}.
The notion of moderately-fast function (defined in
Section~\ref{sec:preliminaries}) will be used to
govern the number of colors used by the coloring algorithms.

\begin{theorem}\label{thm:tight}
Let $\Gamma$ and $\Lambda$ be two related sets of non-decreasing graph-parameters and
let $\cA^{\Gamma}$ be a $g(\tilde{\Delta})$-coloring algorithm with running time bounded
with respect to $\Lambda$ by some function $f$. If
  \begin{enumerate}
  \item there exists a sequence-number function  $s_f$ for $f$;
  \item $g$ is moderately-fast;
  \item the dependence of $f$ on $m$ is bounded by a
    polylog; and
  \item the dependence of $f$ on $\Delta$ is
    moderately-slow;
  \end{enumerate}
then there exists a uniform $\bigo{g(\Delta)}$-coloring algorithm
running in time  $\bigo{f(\Lambda^*)\cdot s_{f}(  f(\Lambda^*))}$.
\end{theorem}

\begin{proof}
Our first goal is to obtain a coloring algorithm that does not require $m$ (and thus requires only
$\Delta$). For this purpose we first define the following problem.

\noindent{The \emph{strong list-coloring} (SLC) problem:}
a configuration for the SLC problem is a pair $(G,\inp)\in \cF(\cal{G})$
such that
\begin{enumerate}[label=\textup{(\arabic{*})}, ref={(\arabic{*})}]
  \item\label{SLC:1} there exists an integer $\hat{\Delta}$ in $\cap_{v\in
    V(G)}\inp(v)$ such that $\hat{\Delta}\ge\Delta$; and
  \item\label{SLC:2} the input $\inp(v)$ of every vertex $v\in V(G)$ contains
a list $L(v)$ of colors contained in $\icc{1}{g(\hat{\Delta})} \times
\icc{1}{\hat{\Delta}+1}$ such that
\[
\forall k\in\icc{1}{g(\hat{\Delta})},\quad\abs{\sst{j}{(k,j)\in L(v)}}\geq \deg_G(v)+1.
\]
\end{enumerate}
Given a configuration $(G,\inp)\in \cF(\cal{G})$, an output vector
$\out$ is a \emph{solution to SLC} if it forms a coloring  and if $\out(v)\in L(v)$
for every node $v\in V(G)$. Condition~\ref{SLC:1} above implies that a local
algorithm for SLC can use an upper bound on $\Delta$, which is the same for
all nodes.
Informally, Condition~\ref{SLC:2} above implies that the list $L(v)$ of colors
available for each node $v$ contains $\deg_G(v)+1$ copies of each color in
the range $\icc{1}{g(\hat{\Delta})}$.

We now design a pruning algorithm $\cP$ for SLC.
Consider a triplet
$(G,\inp,\hout)$, where $(G,\inp)$ is a configuration for SLC and $\hout$ is
some tentative assignment of colors. The set $W$ of nodes to be pruned is
composed of all nodes $u$ satisfying
$\hout(u) \in L(u)$ and $\hout(u) \neq
\hout(v)$ for all $v\in N_G(u)$. For each node $u\in V\setminus W$,
set
\[
L'(u) = L(u)\setminus\sst{\hout(v)}{v\in N_G(u)\cap W}.
\]
In other words, $L'(u)$ contains all the colors in $L(u)$ that are not
assigned to a neighbor of $u$ belonging to $W$.
Algorithm~$\cP$~returns the configuration $(G',\inp')$, where $G'$ is the
subgraph of $G$ obtained by removing the nodes in $W$ and
\[
\inp'(u) = (\inp(u)\setminus L(u))\cup L'(u),
\quad\text{for $u\in V\setminus W$.}
\]

Observe  that if we
start with a configuration $(G,\inp)$ for SLC, then the output $(G',\inp')$
of the pruning algorithm $\cP$ is also a configuration for SLC. Indeed,
for every node $v$ and every integer $k$,
at most $\deg_W(v)$ pairs $(k,j)$ are removed from the
list $L(v)$ of $v$, where $\deg_W(v)$ is the number of neighbors of $v$ that
belong to $W$. On the other hand, the degree of $v$ in $G'$ is reduced by
$\deg_W(v)$. Note also that the input vector of all nodes still contain
$\hat{\Delta}$, which is an upper bound for the maximum degree of $G'$.

Starting with $\cA^\Gamma$, it is straightforward to design a local
algorithm $\cB^{\Gamma'}$ for $SLC$ that depends on $\Gamma'=\Gamma\setminus\{\Delta\}$.
Specifically, $B^{\Gamma'}$ executes $A^{\Gamma}$ using
the good guess $\tilde{\Delta}=\hat{\Delta}$ for the parameter $\Delta$.
Furthermore, if $A^{\Gamma}$ outputs at $v$ a color $c$, then
$B^{\Gamma'}$ outputs the color $(c,j)$ where $j=\min\sst{s}{(c,s)\in L(v)}$.

Given an instance for SLC, we view $\hat{\Delta}$ as a
non-decreasing parameter, and
convert $\Lambda$ to a new set of non-decreasing parameters $\Lambda'$ by replacing $\Delta$
with $\Delta'$. Formally,
if $\Delta\in \Lambda$ then set
$\Lambda'=(\Lambda\setminus \Delta) \cup\hat{\Delta}$, and otherwise, set
$\Lambda'=\Lambda$.
Since $\Gamma$ and $\Lambda$ contain only non-decreasing graph-parameters---and since
$\hat{\Delta}$ is contained in all the inputs---we deduce that the pruning algorithm $\cP$
is $(\Gamma'\cup \Lambda')$-monotone.

Now, we apply Theorem~\ref{thm:parameters} to Algorithm $\cB^{\Gamma'}$,
the sets $\Gamma'$ and $\Lambda'$ of non-decreasing parameters
and the aforementioned pruning algorithm $\cP$ for SLC.
We obtain a uniform algorithm
$\cB$ for SLC and $\cF(\cal{G})$, whose running time is
$\bigo{f(\Lambda'^*)\cdot s_{f}(  f(\Lambda'^*))}$.

We are ready to specify the desired uniform $\bigo{g(\Delta)}$-coloring algorithm.
We define inductively a sequence $(D_i)_{i\in\Natural}$ by setting
$D_1=1$ and
\[
D_{i+1}=\min\sst{\ell}{g(\ell)\geq2g(D_{i})}
\]
for $i\ge1$. As $g$ is moderately-increasing, there is a constant
$\alpha$ such that for each integer $i\ge1$,
\begin{enumerate}
    \item\label{p1} $D_{i+1}\ge\alpha D_i$ and
    \item\label{p2} $g(D_{i+1})\le\alpha^{\log\alpha}g(D_i)$.
\end{enumerate}

Given an initial configuration $(G,\inp)$, we partition
it according to the node degrees. For $i\in\Natural$, let $G_i$ be the
subgraph of $G$ induced by the set of nodes $v$ of $G$ with
$\deg_G(v)\in\icc{D_i}{D_{i+1}-1}$. Let $\inp_i$ be the input $\inp$
restricted to the nodes in $G_i$. The configuration
$(G_i,\inp_i)$, which belongs to $\cF(\cal{G})$, is referred to as \emph{layer} $i$.
Note that nodes can figure out locally which layer they belong to.
Observe also that $D_{i+1}-1$ is an upper bound on node degrees in
layer $i$.

The algorithm proceeds in two phases. In the first phase,
each node in layer $i$ is assigned the list of colors $L''_i=\icc{1}{g(D_{i+1})} \times
\icc{1}{D_{i+1}+1}$, and the degree estimation  $\hat{\Delta}_i= D_{i+1}$. Each layer is now an instance of SLC and we
execute Algorithm $\cB$ in parallel on all layers.
If Algorithm $\cB$ assigns a color $(c,j)$ to a node $v$ in layer $i$ then we change this color to $(g(D_{i+1})+c,j)$.
Hence, for each $i$, layer $i$ is colored with colors taken from $L'_i=\icc{g(D_{i+1})+1}{2g(D_{i+1})} \times
\icc{1}{D_{i+1}+1}$.

Note that nodes in different layers have disjoint color lists, and hence we
obtain a coloring of the whole graph $G$. The number of colors in  $L'_i$ is
at most $2D_{i+1}g(D_{i+1})$.
Let $i_{\max}$ is the maximal integer $i$ such that layer $i$ is
non-empty. The total number of colors used in the first phase is at most
$2D_{i_{\max}+1}g(D_{i_{\max}+1})$, which is $\bigo{\Delta g(\Delta)}$ by Properties~\ref{p1}
and~\ref{p2} above.

Furthermore, the running time of the first
phase of the algorithm is dominated by the running time of the algorithm on
layer~$i_{\max}$.
That is, the running time is at most $\bigo{f(\Lambda'^*)\cdot s_{f}(
f(\Lambda'^*))}$, where $\Lambda'^*$ is the collection of correct parameters
in $\Lambda'$ for layer~$i_{\max}$. Since $D_{i_{\max}+1}=\bigo{\Delta}$ and
the dependence of $f$ on $\Delta$ is moderately-slow, we infer that
$f(\Lambda'^*)=\bigo{f(\Lambda^*)}$.
As $s_f$ is moderately-slow too (by the definition),
we deduce that the running time is $\bigo{f(\Lambda^*)\cdot s_{f}(f(\Lambda^*))}$.

The second phase consists of running a second algorithm  to change the set of possible colors of nodes in
layer $i$ from $L'_i$ to $L_i=\icc{g(D_{i+1})+1}{2
g(D_{i+1})}$.
Specifically, on layer $i$, we execute $\cA^{\Gamma}$ using the guess
$\tilde{\Delta}=D_{i+1}$ for the parameter $\Delta$ and the guess
$\tilde{m}=2D_{i+1}g(D_{i+1})$ for the parameter $m$ (recall that
$\Gamma\subseteq\{\Delta, m\}$).  This procedure colors each layer with colors taken from the range
$\icc{1}{g(D_{i+1})}$.
Let $v$
be in layer $i$ and let $c(v)$ be the color assigned to $v$ by $\cA^{\Gamma}$.
The final color of $v$ given by our desired algorithm $\cA$ is
$g(D_{i+1})+c(v)$. Thus, the colors assigned to
the nodes in layer $i$ belong to $\icc{g(D_{i+1})+1}{2g(D_{i+1})}$.
Therefore, nodes in different layers are assigned distinct colors.
The algorithm is executed on each layer independently, all in parallel.
Hence, we obtain a coloring of $G$.
Moreover, since $g$ is moderately-increasing, the total number of colors
used is $\bigo{ g(\Delta)}$.

Recall that
$D_{i+1}=\bigo{\Delta}$ and $g(D_{i+1})=\bigo{g(\Delta)}$ for all $i$ such
that $G_i$ is not empty.
Hence, we deduce that  the running time of the second phase of the algorithm is bounded from
above by the running time of $\cA^{\Gamma}$ on $(G,\inp)$ using
the guesses $\tilde{\Delta}=\bigo{\Delta}$ and $\tilde{m}=\bigo{\Delta g(\Delta)}$.
Moreover, the
fact that $g(x)$ is bounded by a polynomial in $x$ implies that  $\tilde{m}$
is at most polynomial in $\Delta$, and hence in $m$.

Now, as the dependence of $f$ on $\Delta$ is moderately-slow and
the dependence of $f$ on $m$ is polylogarithmic,
the running time of the second phase of $\cA$ is $\bigo{f(\Lambda)}$.
Combining~this with the running time of the first phase concludes the proof.
\qed
\end{proof}

By Observation~\ref{obs:sequence}, the constant function $s_f=1$ is a
sequence-number function for every additive function~$f$.
Hence, Corollary~\ref{cor:main}\ref{thm:uniformtradeoff} directly follows from
Theorem~\ref{thm:tight}. Regarding edge-coloring, observe that Barenboim
and  Elkin~\cite{BaEl11} obtain their algorithm for general
graphs by running a vertex-coloring algorithm $\cA$ on the line-graph of the given
graph. This algorithm $\cA$ uses $m$ and $\Delta$ in
and the number of colors and time complexity of the
resulting edge-coloring algorithm are that of $\cA$.
Using Theorem~\ref{thm:tight}, one can
transform the algorithm $\cA$ designed for the
family of line graphs into a uniform one, having asymptotically the same number
of colors and running time. Hence, Theorem \ref{cor:main}\ref{thm:edge-coloring} follows.

Let $f\colon \Natural^2\to \Reals$ be given by $f(x_1,x_2)= f_1(x_1)\cdot f_2(x_2)$, where $f_1$ and $f_2$ are
ascending functions.
By Observation~\ref{obs:sequence}, the function
$s_f(i)=\lceil\log i\rceil+1$ is a sequence-number function for
$f$. Therefore, Corollary~\ref{cor:main}\ref{thm:uniform-sublinear} now follows by applying
Theorem~\ref{thm:tight} to the coloring algorithms of Barenboim and Elkin~\cite{BaEl10a}.

\section{Conclusion and Further Research}
\subsection{Pruning Algorithms}
This paper focuses on removing assumptions concerning global knowledge in the context of local algorithms.
We provide transformers taking
a non-uniform local algorithm as a black box and producing a uniform algorithm running in asymptotically the same number of rounds. This is established via the notion of pruning algorithms.
We believe that this novel notion  is of independent interest and
 can be used for other purposes too, e.g., in the context of fault tolerance or dynamic settings.

We remind the reader that we restricted the running time of a pruning algorithm to
be constant.
This is because in all our applications we use constant time pruning algorithms. In fact,
our transformers extend to the case where the given \emph{uniform} pruning
algorithm $\cP$
\begin{itemize}
    \item has running time bounded with respect to a set $\cal{S}$ of non-decreasing
parameters by a (non-decreasing) function $h$; and
    \item is $\cal{S}$-monotone.
\end{itemize}
However, the transformer may incur an additive overhead in the running time of the obtained uniform
algorithms, as these repeatedly use $\cP$. Specifically,
the overhead will be $h(\cal{S}^*)$ times the number of iterations used by the
transformer, which is typically logarithmic in the running time of the non-uniform algorithm.
It would be interesting to have an example of a problem that admits a fast non-trivial uniform
pruning algorithm but does  not admit a  constant time one.

\subsection{Bounded Message Size}
This paper focuses on the $\cal{LOCAL}$ model,
which does not restrict the number of bits used in messages.
Ideally, messages should be short, i.e., using $\bigo{\log n}$ bits.
We found it difficult to obtain a general transformer that takes an arbitrary
non-uniform algorithm using short messages and produces a uniform one
having asymptotically the same running time and message size.
The reason is that techniques similar to those used in this paper, require
guesses that fit for both the
function bounding the running time and the function bounding the message size.
Nevertheless, maintaining the same message size may still be possible given particular
non-uniform algorithms that use messages whose content does not depend on the guessed upper bounds,
such as algorithms that encode in the messages only identifiers, colors, or degrees.

\subsection{Coloring}
Recall that one of the difficulties in obtaining a pruning algorithm for  coloring problems lies in
the fact that the gluing property may not hold, that is,
a pruned node $v$ with color $c$ may have a non-pruned neighbor  $u$
which is  also colored~$c$ in some
correct coloring of the non-pruned subgraph.
In the context of running in iterations, in which one invokes a pruning algorithm
and subsequently, an algorithm~$\cA$ on the non-pruned subgraph (similarly to
Theorem~\ref{thm:parameters}), the aforementioned undesired phenomenon could be prevented
 if the algorithm $\cA$ would avoid coloring node $u$ with color $c$.
With this respect,
we believe that it would be interesting to investigate connections between $g$-coloring problems and
\emph{strong} $g$-coloring problems, in which
each node $v$ is given as an input a list of (forbidden) colors $F(v)$. In a correct solution, each node $v$ must color itself  with a color not in $L(v)$  such that the final configuration is a coloring
using at most $g$ colors.

Finally, recall that our transformer for coloring applies to deterministic
algorithms only. It would be interesting to design a general transformer
that takes non-uniform randomized coloring algorithms (e.g., the ones by
Schneider and Wattenhofer~\cite{ScWa10a}) and transforms them into uniform ones
with asymptotically the same running time.

\begin{acknowledgements}
The authors thank Boaz Patt-Shamir  and the anonymous referees for their careful reading and thoughtful
suggestions. Their comments helped to considerably improve the presentation of the paper.
\end{acknowledgements}

\bibliographystyle{spmpsci}      
\bibliography{color-template-dist}   

\begin{thebibliography}{10}
\providecommand{\url}[1]{{#1}}
\providecommand{\urlprefix}{URL }
\expandafter\ifx\csname urlstyle\endcsname\relax
  \providecommand{\doi}[1]{DOI~\discretionary{}{}{}#1}\else
  \providecommand{\doi}{DOI~\discretionary{}{}{}\begingroup
  \urlstyle{rm}\Url}\fi

\bibitem{ABI86}
Alon, N., Babai, L., Itai, A.: {A fast and simple randomized parallel algorithm
  for the maximal independent set problem.}
\newblock J. Algorithms \textbf{7}, 567--583 (1986)

\bibitem{Awe85}
Awerbuch, B.: Complexity of network synchronization.
\newblock J. ACM \textbf{32}, 804--823 (1985)

\bibitem{AGLP89}
Awerbuch, B., Luby, M., Goldberg, A.V., Plotkin, S.A.: Network decomposition
  and locality in distributed computation.
\newblock In: Proc. 30th IEEE Symp. Found. Comput. Sci. (FOCS), pp. 364--369
  (1989)

\bibitem{BaEl09}
Barenboim, L., Elkin, M.: Distributed $({\Delta}+ 1)$-coloring in linear (in
  {$\Delta$}) time.
\newblock In: Proc. 41st ACM Symp. Theor. Comput. (STOC), pp. 111--120 (2009)

\bibitem{BaEl10a}
Barenboim, L., Elkin, M.: Deterministic distributed vertex coloring in
  polylogarithmic time.
\newblock In: Proc. 29th ACM Symp. Principles Distrib. Comput. (PODC), pp.
  410--419 (2010)

\bibitem{BaEl10b}
Barenboim, L., Elkin, M.: Sublogarithmic distributed mis algorithm for sparse
  graphs using nash-williams decomposition.
\newblock Distrib. Comput. \textbf{22}(5-6), 363--379 (2010)

\bibitem{BaEl11}
Barenboim, L., Elkin, M.: Distributed deterministic edge coloring using bounded
  neighborhood independence.
\newblock In: Proc. 30th ACM Symp. Principles Distrib. Comput. (PODC) (2011)

\bibitem{BeYa76}
Bentley, J.L., Yao, A.C.C.: An almost optimal algorithm for unbounded
  searching.
\newblock Information Processing Lett. \textbf{5}(3), 82--87 (1976)

\bibitem{CFI+08}
Cohen, R., Fraigniaud, P., Ilcinkas, D., Korman, A., Peleg, D.: Label-guided
  graph exploration by a finite automaton.
\newblock ACM Trans. Algorithms \textbf{4}, 42:1--42:18 (2008)

\bibitem{CoVi86}
Cole, R., Vishkin, U.: Deterministic coin tossing and accelerating cascades:
  micro and macro techniques for designing parallel algorithms.
\newblock In: Proc. 18th ACM Symp. Theor. Comput. (STOC), pp. 206--219 (1986)

\bibitem{BGPV08}
Derbel, B., Gavoille, C., Peleg, D., Viennot, L.: On the locality of
  distributed sparse spanner construction.
\newblock In: Proc. 27th ACM Symp. Principles Distrib. Comput. (PODC), pp.
  273--282 (2008)

\bibitem{DePe10}
Dereniowski, D., Pelc, A.: Drawing maps with advice.
\newblock In: Proc. 24th Int. Symp. Distrib. Comput. (DISC), pp. 328--342.
  Springer-Verlag (2010)

\bibitem{FGIP09}
Fraigniaud, P., Gavoille, C., Ilcinkas, D., Pelc, A.: Distributed computing
  with advice: information sensitivity of graph coloring.
\newblock Distrib. Comput. \textbf{21}, 395--403 (2009)

\bibitem{FIP10}
Fraigniaud, P., Ilcinkas, D., Pelc, A.: Communication algorithms with advice.
\newblock J. Comput. Syst. Sci. \textbf{76}, 222--232 (2010)

\bibitem{FKL07}
Fraigniaud, P., Korman, A., Lebhar, E.: Local mst computation with short
  advice.
\newblock In: Proc. 19th ACM Symp. Parallelism Algo. Archit. (SPAA), pp.
  154--160 (2007)

\bibitem{FKP11}
Fraigniaud, P., Korman, A., Peleg, D.: Local distributed decision.
\newblock Submitted for Publication

\bibitem{GoPl87}
Goldberg, A.V., Plotkin, S.A.: Efficient parallel algorithms for $({\Delta}+
  1)$-coloring and maximal independent set problem.
\newblock In: Proc. 19th ACM Symp. Theor. Comput. (STOC), pp. 315--324 (1987)

\bibitem{GPS88}
Goldberg, A.V., Plotkin, S.A., Shannon, G.E.: Parallel symmetry-breaking in
  sparse graphs.
\newblock SIAM J. Discrete Math. \textbf{1}(4), 434--446 (1988)

\bibitem{HKP01}
Ha{{\'n}}{{\'c}}kowiak, M., Karo{{\'n}}ski, M., Panconesi, A.: On the
  distributed complexity of computing maximal matchings.
\newblock SIAM J. Discrete Math. \textbf{15}(1), 41--57 (electronic) (2001/02)

\bibitem{KoKu07}
Korman, A., Kutten, S.: Distributed verification of minimum spanning trees.
\newblock Distrib. Comput. \textbf{20}, 253--266 (2007)

\bibitem{KKP10}
Korman, A., Kutten, S., Peleg, D.: Proof labeling schemes.
\newblock Distrib. Comput. \textbf{22}, 215--233 (2010)

\bibitem{Kuh09}
Kuhn, F.: Weak graph colorings: distributed algorithms and applications.
\newblock In: Proc. 21st ACM Symp. Parallelism Algo. Archit. (SPAA), pp.
  138--144 (2009)

\bibitem{KMW04}
Kuhn, F., Moscibroda, T., Wattenhofer, R.: What cannot be computed locally!
\newblock In: Proc. 23rd ACM Symp. Principles Distrib. Comput. (PODC), pp.
  300--309 (2004)

\bibitem{KuWa06}
Kuhn, F., Wattenhofer, R.: On the complexity of distributed graph coloring.
\newblock In: Proc. 25th ACM Symp. Principles Distrib. Comput. (PODC), pp.
  7--15 (2006)

\bibitem{KuPe00}
Kutten, S., Peleg, D.: Tight fault locality.
\newblock SIAM J. Comput. \textbf{30}(1), 247--268 (electronic) (2000)

\bibitem{LOW08}
Lenzen, C., Oswald, Y., Wattenhofer, R.: What can be approximated locally?:
  case study: dominating sets in planar graphs.
\newblock In: Proc. 20th ACM Symp. Parallelism Algo. Archit. (SPAA), pp. 46--54
  (2008)

\bibitem{Lin87}
Linial, N.: Distributive graph algorithms global solutions from local data.
\newblock In: Proc. 28th IEEE Symp. Found. Comput. Sci. (FOCS), pp. 331--335
  (1987)

\bibitem{Lin92}
Linial, N.: Locality in distributed graph algorithms.
\newblock SIAM J. Comput. \textbf{21}, 193 (1992)

\bibitem{LPR09}
Lotker, Z., Patt-Shamir, B., Ros{{\'e}}n, A.: Distributed approximate matching.
\newblock SIAM J. Comput. \textbf{39}(2), 445--460 (2009)

\bibitem{Lub86}
Luby, M.: A simple parallel algorithm for the maximal independent set problem.
\newblock SIAM J. Comput. \textbf{15}, 1036--1053 (1986)

\bibitem{NaOl02}
Nakano, K., Olariu, S.: Uniform leader election protocols for radio networks.
\newblock IEEE Trans. Parallel Distrib. Syst. \textbf{13}(5), 516--526 (2002)

\bibitem{NaSt95}
Naor, M., Stockmeyer, L.: {What can be computed locally?}
\newblock SIAM J. Comput. \textbf{24}(6), 1259--1277 (1995)

\bibitem{PaRi01}
Panconesi, A., Rizzi, R.: Some simple distributed algorithms for sparse
  networks.
\newblock Distrib. Comput. \textbf{14}, 97--100 (2001)

\bibitem{PaSr96}
Panconesi, A., Srinivasan, A.: On the complexity of distributed network
  decomposition.
\newblock J. Algorithms \textbf{20}(2), 356--374 (1996)

\bibitem{Pel00}
Peleg, D.: {Distributed computing. A locality-sensitive approach.}
\newblock {SIAM Monographs on Discrete Mathematics and Applications, 5.
  Philadelphia, PA: SIAM, Society for Industrial and Applied Mathematics. xvi,
  343 p. } (2000)

\bibitem{ScWa10a}
Schneider, J., Wattenhofer, R.: A new technique for distributed symmetry
  breaking.
\newblock In: Proc. 29th ACM Symp. Principles Distrib. Comput. (PODC), pp.
  257--266 (2010)

\bibitem{ScWa10b}
Schneider, J., Wattenhofer, R.: An optimal maximal independent set algorithm
  for bounded-independence graphs.
\newblock Distrib. Comput. \textbf{22}(5-6), 1--13 (2010)

\bibitem{SzVi93}
Szegedy, M., Vishwanathan, S.: Locality based graph coloring.
\newblock In: Proc. 25th ACM Symp. Theor. Comput. (STOC), pp. 201--207 (1993)

\end{thebibliography}
\end{document}